\documentclass{birkmult}

\usepackage{amsmath, amssymb}
\usepackage{amsthm , amsfonts, latexsym}

\usepackage[mathscr]{euscript}
\usepackage{algorithm,algorithmic}
\usepackage{graphics,graphicx}
\usepackage{enumerate}
\usepackage{cleveref}
\usepackage{color}

\makeatletter
\def\weitdacherl#1{\mathop{\vbox{\m@th\ialign{##\crcr\noalign{\kern3\p@}%
      $\hfil{\scriptscriptstyle\frown}\hfil$\crcr\noalign{\kern1\p@\nointerlineskip}%
      $\displaystyle{#1}$\crcr}}}}
\def\weithutzl#1{\mathop{\vbox{\m@th\ialign{##\crcr\noalign{\kern3\p@}%
      $\hfil{\scriptscriptstyle\smile}\hfil$\crcr\noalign{\kern1\p@\nointerlineskip}%
      $\displaystyle{#1}$\crcr}}}}
 
\def\part{\@startsection{part}{0}%
  \z@{2.5\linespacing\@plus\linespacing}{1.0\linespacing}%
  {\LARGE\bfseries\raggedright}}
\makeatother
\newtheorem{thm}{Theorem}[section]
\newtheorem{cor}[thm]{Corollary}
\newtheorem{lem}[thm]{Lemma}
\newtheorem{prop}[thm]{Proposition}
\theoremstyle{definition}
\newtheorem{defn}[thm]{Definition}
\theoremstyle{remark}

\numberwithin{equation}{section}
\newcommand\dirmin{\ensuremath{\weithutzl{\times}}\,}
\newcommand\dirmax{\ensuremath{\weitdacherl{\times}}\,}
\newcommand\dirnon{\ensuremath{\widetilde{\times}}\,}

\newcommand\strmin{\ensuremath{\weithutzl{\boxtimes}}\,}
\newcommand\strmax{\ensuremath{\weitdacherl{\boxtimes}}\,}

\newcommand*{\bigtimes}{\mathop{\hbox{\Large{$\times$}}}}

\newcommand{\la}{\langle}
\newcommand{\ra}{\rangle}

\newcommand{\of}{\overrightarrow{f}}
\renewcommand{\oe}{\overrightarrow{e}}
\newcommand{\oE}{\overrightarrow{E}}
\newcommand{\oH}{\overrightarrow{H}}
\newcommand{\oN}{\overrightarrow{N}}

\newcommand\disc{\text{disc}\,}
\newcommand\timesR{\ensuremath{\fbox{\text{\scriptsize r}}}\,}
\newcommand\timesH{\ensuremath{\circledcirc}\,}

\newcommand{\categ}{\mathop{\hbox{\scriptsize{$\blacksquare$}}}}

\crefname{thm}{}{}

\begin{document}
\sloppy
%
%
%
%
%
%
%
%
%
\title[A Survey on Hypergraph Products]
 {A Survey on Hypergraph Products}
\author{Marc Hellmuth}

\address{%
Center for Bioinformatics \\
Saarland University \\
Building E 2.1, Room 413 \\
P.O. Box 15 11 50 \\
D - 66041 Saarbr\"{u}cken \\
Germany }

\email{marc@bioinf.uni-leipzig.de}

\author{Lydia Ostermeier}
\address{Max Planck Institute for Mathematics in the Sciences\\
Inselstrasse 22, \\ D-04103 Leipzig, \\ Germany\\[0.2cm]
Bioinformatics Group, \\ Department of Computer Science and
Interdisciplinary Center for Bioinformatics \\
University of Leipzig,\\
H{\"a}rtelstrasse 16-18, D-04107 Leipzig, Germany }
\email{glydia@bioinf.uni-leipzig.de}

\author{Peter F. Stadler}
\address{
Bioinformatics Group, \\
Department of Computer Science; and
Interdisciplinary Center for Bioinformatics,\\
University of Leipzig,\\
H{\"a}rtelstrasse 16-18, D-04107 Leipzig,Germany\\[0.1cm]
Max Planck Institute for Mathematics in the Sciences\\
Inselstrasse 22, D-04103 Leipzig, Germany\\[0.1cm]
RNomics Group, Fraunhofer Institut f{\"u}r Zelltherapie und Immunologie,
Deutscher Platz 5e, D-04103 Leipzig, Germany\\[0.1cm]
Department of Theoretical Chemistry,  University of Vienna,
  W{\"a}hringerstra{\ss}e 17, A-1090 Wien, Austria\\[0.1cm]
Santa Fe Institute, 1399 Hyde Park Rd., Santa Fe, NM87501, USA}
\email{studla@bioinf.uni-leipzig.de}

\thanks{This work was supported in part by the \emph{Deutsche
    Forschungsgemeinschaft} (DFG) Project STA850/11-1 
    within the EUROCORES Programme EuroGIGA
  (project GReGAS) of the European Science Foundation.}

\subjclass{Primary 99Z99; Secondary 00A00}

\keywords{Hypergraph invariants, products, set systems}

\date{22 Dec 2011}

\dedicatory{\begin{flushleft}	{\bf Erratrum:} In the accepted version of this survey \cite{recursiveSurvey11}
				it is mistakenly stated that the direct products
				$\dirmax$ and $\dirnon$ and the strong product $\strmax$ are associative.
				In \cite{HHOS:16}, we gave  counterexamples for these cases and proved
				associativity of the hypergraph products $\dirmin,\strmin$.  \end{flushleft}
}

\begin{abstract}
	 A surprising diversity of different products of hypergraphs have been
  discussed in the literature.  Most of the hypergraph products can be
    viewed as generalizations of one of the four standard graph
    products. The most widely studied variant, the so-called square
    product, does not have this property, however.  Here we
  survey the literature on hypergraph products with an emphasis on
  comparing the alternative generalizations of graph products and the
  relationships among them.  In this context the so-called 2-sections and
  L2-sections are considered. These constructions are closely linked to
  related colored graph structures that seem to be a useful tool for the
  prime factor decompositions w.r.t.\ specific hypergraph products.  We
  summarize the current knowledge on the propagation of hypergraph
  invariants under the different hypergraph multiplications.  While the
  overwhelming majority of the material concerns finite (undirected)
  hypergraphs, the survey also covers a summary of the few results on
  products of infinite and directed hypergraphs.
\end{abstract}

\maketitle

\part{Introduction}

There are only four ``standard graph products'' that preserve the salient
structure of their factors and behave in an algebraically reasonable
way. Their structural features have been studied extensively over the last
decades. It is well known how many of the important graph invariants
propagate under product formation, and efficient algorithms have been
devised to decompose graph products into their prime factors. Several
monographs cover the topic in substantial detail and serve as standard
references \cite{IMKL-00,IMKLDO-08,Hammack:2011a}.

In contrast, very little is known about product structures of hypergraphs,
even though hypergraphs have become increasingly important models of
network structures. Here we survey the existing literature, focusing on the
basic properties of the various hypergraph products and their mutual
relationships. In this introductory part we will first investigate in which
sense the standard graph products have distinguished properties. After
introducing the necessary notation, and defining the most interesting
hypergraph invariants we proceed to discuss a set of desirable properties
of hypergraph products that generalize the situation in graphs.  Much
  of the published literature is concerned with the so-called square
  product, which does \emph{not} arise as a natural generalization of a
  graph product. Most of the other constructions, albeit less well
  investigated so far, can be described as generalizations of a
  corresponding graph product. The link between hypergraph products is also
  stressed by constructions such as 2-sections and L2-sections
  \cite{Berge:Hypergraphs, BrettoSilvestre10:Factorization}. We therefore
  choose to emphasize the generalizations of graph products in our
  survey.  The following sections are then concerned with a review of the
literature on the individual notions of hypergraph products. The
  literature is complemented by several new results that bridge some of the
  obvious gaps in particular for the rarely studied products. Our survey
also includes a complementary recursive exposition of several new
constructions and their basic properties \cite{recursiveSurvey11}.

\section{Graph Products}

Graph products are natural structures in discrete mathematics
\cite{Ham-09,OHK+09} that arise in a variety of different
contexts, from computer science \cite{AMA-07,HJH+10, Hel-11} and computational
engineering \cite{KK-08,KR-04} to theoretical biology
\cite{Fontana:98a,Fontana:98b,Cupal:00a,BMRStadler:01a,Wagner:03a}.  
In this section we briefly outline the commonly
investigated graph products and their most salient properties to provide a
frame of reference for our subsequent discussion of hypergraph products.

We consider only finite and undirected graphs $G=(V,E)$ with non-empty
vertex set $V$ and edge set $E$.  A graph is \emph{non-trivial} if it has at
least two vertices. Graph products can be constructed in many different
ways. For example, different constructions arise depending on whether loops
are considered or not. There are, however, three basic properties that are
required for any meaningful definition of a graph product:
\begin{itemize}
\item[(P1)] The vertex set of a product is the Cartesian product
  of the vertex sets of the factors.
\item[(P2)] Adjacency in the product depends on the adjacency properties
  of the projections of pairs of vertices into the factors.
\item[(P3)] The product of a simple graph is a simple graph.
\end{itemize}
As shown in \cite{IMIZ-75}, there are $256$ different possibilities to
define a graph product satisfying (P1), (P2), and (P3). Only six of them 
are commutative, associative and have a unit. Only four products 
satisfy the following additional condition:
\begin{itemize}
\item[(P4)] At least one of the projections of a product onto its factors 
  is a so-called weak homomorphism (edges are mapped to edges or to vertices).
\end{itemize}
These four products are known as the \emph{standard graph products}
\cite{IMKL-00,Hammack:2011a}: the \emph{Cartesian} product $\Box$,
the \emph{direct} product $\times$, the \emph{strong} product $\boxtimes$, and
the \emph{lexicographic} product $\circ$.

In all products the vertex set $V(G_1\circledast G_2)$ is defined as the
Cartesian product $V(G_1)\times V(G_2)$, 
$\circledast \in \{\Box,\ \times,\ \boxtimes,\ \circ \}$. 
Two vertices $(x_1,x_2)$, $(y_1,y_2)$ are adjacent
in $G_1\boxtimes G_2$ if one of the following conditions is
satisfied:
\begin{itemize} 
 \item[(i)] $(x_1,y_1)\in E(G_1)$ and $x_2=y_2$,
 \item[(ii)]$(x_2,y_2)\in E(G_2)$ and $x_1 = y_1$,
 \item[(iii)] $(x_1,y_1)\in E(G_1)$ and $(x_2,y_2)\in E(G_2)$.
\end{itemize}
In the \emph{Cartesian product} vertices are adjacent if and only if they satisfy
(i) or (ii). Consequently, the edges of a strong product that satisfy (i)
or (ii) are called \emph{Cartesian} edge, the others are the
\emph{non-Cartesian} edges. In the direct product vertices are only
adjacent if they satisfy (iii). Thus, the edge set of the strong product is
the union of edges in the Cartesian and the direct product.  In the
lexicographic product vertices are adjacent if and only if $(x_1,y_1)\in E(G_1)$
or they satisfy (ii). 

Three of these products, the Cartesian, the direct and the strong product
are commutative, associative, and distributive with respect to the disjoint
union. The lexicographic product is associative, not commutative, and only
left-distributive with respect to the disjoint union.  All products have a
unit element, that is the single vertex graph $K_1$ for the Cartesian, the
strong and the lexicographic product, and the single vertex graph with a
loop $\mathscr{L}K_1$ for the direct product.

Connectedness of the products depends on the connectedness of the factors.
The Cartesian and the strong product is connected if and only if all of its
factors are connected. The direct product of non-trivial connected factors 
is connected if and only if at
most one factor is bipartite. The lexicographic product $\circ_{i=1}^n G_i$
is connected if and only if $G_1$ is connected. The \emph{costrong product}
$G_1*G_2$, with edge set $E(G_1\circ G_2) \cup E(G_2\circ G_1)$, can be seen
as a symmetrized version of the lexicographic product. It is also closely
related to the strong produce by virtue of the identity $G_1*G_2 =
\overline{\overline{G_1}\boxtimes \overline{G_2}}$
\cite{GasztImrich71:lexico-englAbstract}.

Connected graphs have a unique \emph{Prime Factor Decomposition} (PFD)
w.r.t.\ the strong and the Cartesian product and connected non-bipartite
graphs have a unique PFD w.r.t.\ the direct product.  The PFD w.r.t.\ the
lexicographic product is unique only under strict conditions w.r.t.\
connectivity properties based on the prime factors \cite{Hammack:2011a}.

\section{Hypergraphs}

\subsection{Basic Definitions}

Hypergraphs are a natural generalization of undirected graphs in which
``edges'' may consist of more than 2 vertices. More precisely, a
\emph{(finite) hypergraph} $H=(V,E)$ consists of a (finite) set $V$ and a
collection $E$ of non-empty subsets of $V$.

The elements of $V$ are called \emph{vertices} and the elements of $E$ are
called \emph{hyperedges}, or simply \emph{edges} of the hypergraph.
Throughout this survey, we only consider hypergraphs without multiple edges
and thus, being $E$ a usual set.  If there is a risk of confusion we will
denote the vertex set and the edge set of a hypergraph $H$ explicitly by
$V(H)$ and $ E(H)$, respectively.

A hypergraph $H=(V,E)$ is \emph{simple} if no edge is contained in any
other edge and $|e|\geq 2$ for all $e\in E$.  The dual $H^*$ of a
hypergraph $H=(V,E)$ is the hypergraph whose vertices and edges are
interchanged, so that $V(H^*)=\{e_i^*\mid e_i\in E\}$ and edge set
$E(H^*)=\{v_i^*\mid v_i\in V\}$ with $v_i^* = \lbrace e_j^* \mid v_i \in
e_j \rbrace $.

For a (simple) hypergraph $H=(V,E)$ let $\mathscr{L}H:=\left(V,E\cup
  \left\{\{x\}\mid x\in V\right\}\right)$ denote the hypergraph which is
formed from $H$ by adding a loop to each vertex of $H$.  Conversely, for a
hypergraph $H'=(V',E')$ let $\mathscr{N}H':= \left(V',E'\setminus
  \left\{\{x\}\mid x\in V'\right\}\right)$ denote the hypergraph which
emerges from $H'$ by deleting all loops.

Two vertices $u$ and $v$ are \emph{adjacent} in $H=(V, E)$ if there is an
edge $e\in E$ such that $u,v\in e$. If for two edges $e,f\in E$ holds
$e\cap f\neq\emptyset$, we say that $e$ and $f$ are \emph{adjacent}.  A
vertex $v$ and an edge $e$ of $H$ are \emph{incident} if $v\in e$.  The
\emph{degree} $\deg(v)$ of a vertex $v\in V$ is the number of edges
incident to $v$. The \emph{maximum degree} $\max_{v\in V} \deg(v)$ is
denoted by $\Delta(H)$.

The \emph{rank} of a hypergraph $H=(V,E)$ is $r(H)=\max_{e\in E}|e|$, the
\emph{anti-rank} is $s(H)=\min_{e\in E}|e|$.  A \emph{uniform hypergraph} $H$ is
a hypergraph such that $r(H)=s(H)$.  A simple uniform hypergraph of rank
$r$ will be called \emph{$r$-uniform}.  A hypergraph with $r(H)\leq 2$ is a
\emph{graph}.  A $2$-uniform hypergraph is usually known 
as a \emph{simple graph}.

A \emph{partial hypergraph} $H'=(V',E')$ of a hypergraph $H=(V,E)$, denoted
by $H'\subseteq H$, is a hypergraph such that $V'\subseteq V$ and
$E'\subseteq E$.  In the class of graphs partial hypergraphs are called
\emph{subgraphs}.  The partial hypergraph $H'=(V',E')$ is \emph{induced} if
$E' = \{e\in E\mid e\subseteq V'\}$.  Induced hypergraphs will
be denoted by $\left\la V'\right\ra$. A partial hypergraph of a simple
hypergraph is always simple. 

A \emph{walk} in a hypergraph $H=(V, E)$ is a sequence $P_{v_0,v_{k}} =
(v_0,e_1,v_1,e_2,\ldots,e_k,v_k)$, where $e_1,\ldots,e_k \in E$ and
$v_0,\ldots,v_k\in V$, such that each $v_{i-1}\neq v_i$ and $v_{i-1},v_i\in
e_i$ for all $i=1,\ldots,k$.  The walk $P_{v_0,v_{k}}$ is said to
\emph{join} the vertices $v_0$ and $v_k$.  A \emph{$p$-path} is a walk
where the vertices $v_0,\ldots,v_k$ are all distinct and for all $r,s\in
\{1,\dots,k\}, r\leq s$ with $e_r=e_s$ follows that $s-r \leq p-1$ and $e_r
= e_{r+1} = \dots = e_s$.  A path between two edges $e_i$ and $e_j$ is any
path $P_{v_i v_j}$ joining vertices $v_i\in e_i$ and $v_j\in e_j$.  A
$1$-path is just called a \emph{path}, i.e., all vertices and all edges are
different.  The minimum number of pairwise vertex and edge disjoint paths
of a hypergraph $H$ whose union contains all vertices of $H$ is called
\emph{vertex path partition number} and will be denoted by $\wp(H)$.
Note, the path partion number satisfies $\wp(H)\leq |V(H)|$,
since there is always a partition of a hypergraph into paths of length $0$.
A \emph{cycle} is a sequence $(v_0,e_1,v_1,e_2,\ldots,v_{k-1},e_k,v_0)$, such
that $P_{v_0,v_{k-1}}$ is a path.  A $p$-path or a cycle is
\emph{Hamiltonian} in $H$ if it contains all vertices of $H$.  The
\emph{length of a path} or \emph{a cycle} is the number of edges contained
in the path or cycle, resp.

The \emph{distance} $d_H(v,v')$ between two vertices $v_0,v_k$ of $H$ is
the length of a shortest path joining them.  We set $d_H(v,v')=\infty$ if
there is no such path.  A hypergraph $H=(V,E)$ is called \emph{connected},
if any two vertices are joined by a path.  A partial hypergraph
$H'\subseteq H$ is called \emph{convex}, if all shortest paths in $H$
between two vertices in $H'$ are also contained in $H'$.

\subsection{Homomorphisms and Covering Constructions}

For two hypergraphs $H_1=(V_1,E_1)$ and $H_2=(V_2, E_2)$ a
\emph{homomorphism} from $H_1$ into $H_2$ is a mapping $\varphi:
V_1\rightarrow V_2$ such that
$\varphi(e)=\{\varphi(v_1),\ldots,\varphi(v_r)\}$ is an edge in $H_2$, if
$e=\{v_1,\ldots,v_r\}$ is an edge in $H_1$.  Note, a homomorphism from
$H_1$ into $H_2$ implies also a mapping $\varphi_{\mathcal{E}} : E_1
\rightarrow E_2$.  A mapping $\varphi: V_1\rightarrow V_2$ is a \emph{weak
  homomorphism} if edges are mapped either on edges or on vertices.

A homomorphism $\varphi$ that is bijective is called an \emph{isomorphism}
if holds $\varphi(e)\in E_2$ if and only if $e\in E_1$.  We say, $H_1$ and $H_2$
are \emph{isomorphic}, in symbols $H_1\cong H_2$ if there exists an
isomorphism between them. An isomorphism from a hypergraph $H$ onto itself
is an \emph{automorphism}.

The hypergraph $H'=(V',E')$ is a \emph{$k$-fold covering} of a hypergraph
$H=(V,E)$ if there is a surjective homomorphism $\pi: H' \rightarrow H$ for
which
\begin{enumerate}
\item $|\pi^{-1}(v)|=|\pi_{\mathcal{E}}^{-1}(e)|=k$ for all $v \in V$,
  $e\in E$ and
\item $e' \cap f' = \emptyset$ for all distinct $e',f'$ in
  $\pi_{\mathcal{E}}^{-1}(e)$, $e\in E$.
\end{enumerate}
$H$ is then called the \emph{quotient hypergraph} of $H'$ and $\pi$ is
called the \emph{covering projection} \cite{Doerfler79:CoversDirectProd}.
If $k=2$, $H'$ is called \emph{double cover}
\cite{Doerfler78:DoubleCovers}.

\subsection{$L2$-sections}

The notion of so-called $L2$-sections has proved to be an extremely useful
tool for hypergraph product recognition algorithms. In the following we
therefore consider the $2$-section and $L2$-section of hypergraphs
\cite{Berge:Hypergraphs, BrettoSilvestre10:Factorization} in some detail.
In the context of $2$-sections and $L2$-sections we will consider only
hypergraphs without loops throughout this survey.

The $2$-section $[H]_2$ of a hypergraph $H=(V,E)$ is the graph $(V,E')$
with $E'=\left\{\{x,y\}\subseteq V\mid x\neq y,\,\exists\; e\in E:
  \{x,y\}\subseteq e\right\}$, that is, two vertices are adjacent in
$[H]_2$ if they belong to the same hyperedge in $H$. Thus, every hyperedge
of $H$ is a clique in $[H]_2$.  Note, the $2$-section $[H]_2$ of a
hypergraph $H=(V,E)$ is only uniquely determined if $H$ is
\emph{conformal}, that is, for every subset $W\subseteq V$ holds that if
$\la W \ra$ is a clique in $[H]_2$ then $W\in E$.

Let $\mathbb P(X)$ denotes the power set of the set $X$.
The $L2$-section $[H]_{L2}$ of a hypergraph $H=(V,E)$ is its $2$-section
together with a mapping $\mathcal{L}:E'\rightarrow\mathbb P(E)$ with
$\mathcal{L}(\{x,y\})=\left\{e\in E\mid \{x,y\}\subseteq e\right\}$.
Usually, the $L2$-section $[H]_{L2}$ is written as the triple
$\Gamma=\left(V,E([H]_2),\mathcal{L}\right)$.  In addition to $2$-sections,
the $L2$-section also provides the possibility to trace back the
information which of the edges of $[H]_2$ is associated to which of the
hyperedges in $H$. 
Thus, the original hypergraph can be reconstructed from its
  $L2$-section.  The inverse $[\Gamma]_{L2}^{-1} = (V,E)$ of an
$L2$-section $\Gamma=(V,E',\mathcal{L})$ is the hypergraph with
$E=\bigcup_{e\in E'} \mathcal{L}(e)$.  Hence, the inverse
$[\Gamma]_{L2}^{-1}$ of an $L2$-section is the hypergraph $H=(V,E)$ that
has $L2$-section $\Gamma$.

Two $L2$-sections 
$\Gamma_1=\left(V_1,E_1,\mathcal{L}_1\right)$ and
$\Gamma_2=\left(V_2,E_2,\mathcal{L}_2\right)$ are isomorphic, in symbols 
$\Gamma_1\cong \Gamma_2$, 
if there is an isomorphism $\varphi$ between the graphs $(V_1,E_1)$ and 
$(V_2,E_2)$ such that $e\in \mathcal{L}_1(\{x,y\})$ if and only if 
$\{\varphi(z)\mid z\in e\}\in L_2(\{\varphi(x),\varphi(y)\})$ 
for all $x,y \in V_1$ and $e\subseteq V_1$.
Every hypergraph is uniquely (up to isomorphism) determined by its 
$L2$-section and \emph{vice versa} 
\cite{Bretto09:HyperCartProd, BrettoSilvestre10:Factorization}, i.e., 
$H\cong H'$ if and only if $[H]_{L2} \cong [H']_{L2}$. 

A very useful property of the $2$-section is the following:
\begin{lem}[Distance Formula]
  \label{lem:distH2}
  Let $H = (V,E)$ be a hypergraph and $x,y\in V$.
  Then the distances between $x$ and $y$
  in $H$ and in $[H]_2$ are the same.
\end{lem}
\begin{proof}
  Note, $x$ and $y$ are in different connected components of $H$ if and
  only if $x$ and $y$ are in different connected components of $[H]_2$ and
  hence, $d_H(x,y) = d_{[H]_2}(x,y) = \infty$.  Thus, w.l.o.g. assume $H$
  (and hence $[H]_2$) to be connected.  Let $P = (x, e_1, v_1, \ldots,
  v_{k-1}, e_k, y)$ denote a shortest path between $x$ and $y$ in $H$. By
  construction of $[H]_2$ there is a walk $P' = (x, e'_1, v_1, \ldots,
  v_{k-1}, e'_k, y)$ in $[H]_2$. Thus, $k=d_H(x,y) \geq d_{[H]_2}(x,y)=l$.
  Assume, $k>l$. Then there is a path $Q' = (x, f'_1, v_1, \ldots, v_{k-1},
  f'_l, y)$ in $[H]_2$. Thus, for all $f'_i$ there is an edge $f_i\in E(H)$
  such that $f_i'\subseteq f_i$ and hence, a walk of length $l$ in $H$, a
  contradiction.
\end{proof}

\subsection{Invariants}

In the following paragraphs we briefly introduce the hypergraph invariants
that are most commonly studied in the context of hypergraph products.
We will assume throughout that $H=(V,E)$ is a given hypergraph.

\subsubsection{Independence, Matching and Cover} 

A set $S\subseteq V$ is \emph{independent} if it contains no edge of $E$;
the maximum cardinality of an independent set is denoted by $\beta(H)$ and
is called the \emph{independence number} of $H$. Some of the older
literature, e.g.\ \cite{Berge:DirectProd, Sterboul74:ChromDirect} use the
term \emph{stable} and \emph{stability number} for this concept.
 
A set $T\subseteq V$ is called a \emph{cover} of H if it intersects every
edge of H, i.e., $T \cap e \neq \emptyset$ for all $e\in E$.  The minimum
cardinality of the covers is denoted by $\tau(H)$, and called the
\emph{covering number} of H.  Cover and covering number are also known as
\emph{transversal} or \emph{transversal number} \cite{Berge:DirectProd}.

A \emph{fractional cover} of $H$ is a mapping $t:V\rightarrow
\mathbb{R}^+_0$ such that $\sum_{v \in e} t(v) \geq 1$ for all $e\in
E$. The value $\min_{t} \sum_{v \in V} t(v)$ over all fractional covers $t$
is called \emph{fractional covering number} and denoted by $\tau^*(H)$.
Note that every cover induces a fractional cover by defining $t(v)=1$ if
$v\in T$ and $t(v)=0$ else \cite{Berge:Hypergraphs}.

A subset $M \subseteq E$ is a \emph{matching} if every pair of edges
from $M$ has an empty intersection.  The maximum cardinality of a matching
$M$ is called the \emph{matching number}, denoted by $\nu(H)$
\cite{Berge:Hypergraphs}. 

The partition number $\rho(H)$ of $H$ denotes the minimal
number of pairwise disjoint edges of $E$ which together cover $V$
if such a partition exists, else we set $\rho(H) = \infty$ 
 \cite{AhlswedeCai97:ExtremalSetPart, AhlswedeCai94:Partitioning}.

\subsubsection{Coloring} 

A \emph{coloring} of a hypergraph $H$ is mapping $c$ from either $V$ or $E$
into a set of colors $C = \{1,\dots, k\}$.  We refer to $c:E\to C$ as an
\emph{edge-coloring} and to $c:V\to C$ as a \emph{vertex-coloring} or
simply \emph{coloring}.

A \emph{proper} coloring of a hypergraph $H$ is a coloring $c: V
\rightarrow C$ such that $\{v \mid c(v)=i\}$ is an independent set for all
$i\in C$.  The \emph{chromatic number} $\chi(H)$ is the minimal number of
colors that admit a proper coloring of $H$.  Hence, the \emph{chromatic
  number} $\chi(H)$ is the minimum number of independent sets $V_1,\cdots,
V_{\chi(H)}$ into which $V$ can be partitioned.  A proper \emph{strong}
coloring of a hypergraph $H$ is a proper coloring such that for all edges
$e\in E$ holds that $c(v)\neq c(w)$ for all distinct vertices $v,w \in e$.
The \emph{strong} chromatic number $\chi_s(H)$ is the minimal number $k$ of
colors that admit a strong $k$-coloring of $H$.

The \emph{($k$-color) discrepancy} of a hypergraph measures the deviation
of a coloring $c$ from a so-called balanced coloring, that is a coloring in
which each hyperedge contains same number of vertices of each color.  More
formally, the \emph{discrepancy of a coloring $c$} and the \emph{$k$-color
  discrepancy of $H = (V,E)$} are defined as follows:
$$\disc(H,c) = \max_{e\in E} \max_{1\leq i \leq k}
\left||c^{-1}(i)\cap e|-\frac{1}{k}|e|\right|$$
and
$$\disc(H,k) = \min_{c:V\rightarrow\{1,\ldots,k\}}\disc(H,c)\,.$$

A \emph{proper} edge-coloring of a hypergraph $H$ is a coloring $c: E
\rightarrow C$ such that $c(e) \neq c(f)$ for all distinct incident edges
$e,f\in E$.  The \emph{chromatic index} $q(H)$ of $H$ is the minimum number
of colors that admit a proper edge-coloring. Clearly, $q(H)\geq \Delta(H)$.
A hypergraph has the \emph{colored hyperedge property} if $q(H) =
\Delta(H)$.

\subsubsection{Helly Property}

For $v\in V$, a \emph{star} $H(v)$ of $H$ with center $v$ is the set of all
edges $e\in E$ such that $v\in e$.  For a given simple hypergraph $H$ a
subset $E'$ of $E$ is an \emph{intersecting family} if every pair of
hyperedges of $E'$ have a non-empty intersection.  A hypergraph has the
\emph{Helly property} if each intersecting family is a star.  An
interesting characterization of Helly hypergraphs can be found in
\cite{Bandelt:91Helly}: A hypergraph has the Helly property if and only if
its dual is conformal.

\section{Basic Properties of Hypergraph Products}

Definitions of hypergraph products, to our knowledge, have never been
compared systematically in a way similar to graph products.  Most of the
hypergraph products can be viewed as a generalization of respective graph
products. However, one of the most studied hypergraph product, the
so-called square product, does not provide this property.  Therefore, it
appears useful to make explicit the desirable properties of hypergraph
products. We begin with the direct generalization of the requirements for
graph products:
\begin{itemize}
\item[(P1)] The vertex set of a product is the Cartesian product
  of the vertex sets of the factors.
\item[(P2)] Adjacency in the product depends only on the adjacency
  properties of the projections of pairs of vertices into the factors.
\item[(P3)] The product of simple hypergraphs is again a simple hypergraph.
\item[(P4)] At least one of the projections of a product onto its factors 
  is a weak homomorphism.
\end{itemize}
Since graphs can be interpreted as the special hypergraphs with 
$|e|\le 2$ for all $e\in E$, we would like to consider hypergraph products
that specialize to graph products:
\begin{itemize}
\item[(P5)] The hypergraph product of two graphs is again a graph.
\end{itemize}
For hypergraphs, these requirements appear to give more freedom than for
graphs. Property (P2) posits that the presence of an edge
$(x_1,x_2)\sim(y_1,y_2)$ must be determined by the presence or absence of
the adjacencies $x_1\sim y_1$ and $x_2\sim y_2$ and a rule deciding whether
$x_1=y_1$ and $x_2=y_2$ is to be treated like an edge or its absence,
leading to $2^8=256$ distinct operations, see \cite{IMIZ-75}. For hypergraphs, however, this
leads only to a restriction on edges but does not provide a complete recipe
for the construction of the edge set of the product. As a consequence, it
is possible to find several non-equivalent generalizations of the standard
graph products as we shall see throughout this survey.

As in the case of usual graph products at least associativity is
desirable. All products that are treated in this survey are associative. We
omit the proofs for this, since they can be done equivalently to the proofs
as in \cite{Hammack:2011a}.  Thus, the hypergraph products of finitely many
factors are well defined and it suffices to prove the results for two (not
necessarily prime) factors only.  Furthermore, all products, except the
lexicographic product are commutative.

Before we proceed with our analysis of hypergraph products, we need to
introduce some specific notations:\newline
Let $\circledast_{i=1}^n H_i = (V,E) = (\bigtimes_{i=1}^n V(H_i),
E(\circledast_{i=1}^n H_i))$ be an arbitrary hypergraph product.  The
\emph{projection} $p_j:V\to V(H_j)$ is defined by $v=(v_1,\dots,v_n)
\mapsto v_j$. We will call $v_j$ the \emph{$j$-th coordinate} of the vertex
$v\in V$.
For a given vertex $w\in V(H)$ the \emph{$H_j$-layer through $w$} is the 
partial hypergraph of $H$   
\[ H_j^w = \left\langle \{v\in V(H) \mid p_k(v)=p_k(w)
  \text{ for } k\neq j \}\right\rangle.
\]
If for a hypergraph product $\circledast_{i=1}^n H_i$ holds
$H_i\cong H$ for all $i=1,\ldots,n$ we will denote this hypergraph
simply by $H^{\circledast n}$. 

Let $U$ denote the unit element, if one exists, of an arbitrary product
$\circledast$, i.e., $H = H \circledast U$ for all hypergraphs $H$.  Since
all hypergraph products considered here have vertex set $V_1\times V_2$ the
unit must always be a hypergraph with a single vertex. A hypergraph is said
to be \emph{prime} if the identity $H = H_1 \circledast H_2$ implies that
$H_1\cong U$ or $H_2 \cong U$. Not all hypergraph products have a unit
element. Prime factors and a prime factor decomposition cannot be
meaningfully defined unless there is a unit.

\part{Cartesian Product}

\section{The Cartesian Product}
The Cartesian product of hypergraphs has been investigated by several
authors since the 1960s \cite{Imrich67:Mengensysteme,
  Imrich70:SchwachKartProd, Bretto09:HyperCartProd, Bretto06:Helly,
  BrettoSilvestre10:Factorization, Doerfler79:CoversDirectProd,
  OstHellmStad11:CartProd, Sonntag89:HamCart}. It is probably the
best-studied construction.

\subsection{Definition and Basic Properties}

\begin{defn}[Cartesian Product of Hypergraphs]
  The \emph{Cartesian product} $H=H_1\Box H_2$ of two hypergraphs $H_1$ and
  $H_2$ has vertex set $V(H)=V(H_1)\times V(H_2)$ and the edge set
  \begin{align*} 
    E (H)=&\big\{\{x\}\times f :x\in V(H_1), f\in E (H_2)\big\}\\
    &\cup \big\{e\times \{y\}:e\in E (H_1), y\in V(H_2)\big\}.
  \end{align*}
\end{defn}

The Cartesian product is associative, commutative, distributive with
respect to the disjoint union and has the single vertex graph $K_1$ as a
unit element \cite{Imrich67:Mengensysteme}.  The Cartesian product of two
simple hypergraphs is a simple hypergraph.  A Cartesian product hypergraph,
furthermore, is connected if and only if all of its factors are connected
\cite{Imrich67:Mengensysteme}.  For the rank and anti-rank, respectively,
of a Cartesian product hypergraph $H_1\Box H_2$ holds:
\begin{align*}
	r(H_1\Box H_2)=\max\{r(H_1),r(H_2)\}\\
	s(H_1\Box H_2)=\min\{s(H_1),s(H_2)\}.
\end{align*}

The projections onto the factors are weak homomorphisms.
According to \cite{Sabidussi60:GraphMult} the Cartesian product of hypergraphs 
can be described in terms of projections as follows:
For $H=H_1\Box H_2$, with $H_i=(V_i,E_i)$, $i=1,2$ and
$e\subseteq V(H)$ we have $e\in E(H)$ if and only if 
  there is an $i\in \{1,2\}$, s.t.
  \begin{itemize}
  \item[(i)] $p_i(e)\in E_i$ and
  \item[(ii)] $|p_j(e)|=1$ for $j\neq i$.
  \end{itemize}
Furthermore, $|p_i(e)|=|e|$.
 
The \emph{$H_j$-layer through $w$} of a Cartesian product $H$ is induced by
all vertices of $H$ that differ from $w\in V(H)$ exactly in the $j$-th
coordinate. Moreover, $H_j^w \cong H_j$.

\subsection{Relationships with Graph Products}

The restriction of the Cartesian product to graphs coincides with the usual
Cartesian graph product. The $2$-sections of hypergraphs are also
well-behaved:
\begin{prop}[\cite{Bretto09:HyperCartProd}]	\label{prop:CartProd_2-sect}
The $2$-section of $H= H' \Box H''$ is the Cartesian product 
of the $2$-section of $H'$ and the $2$-section of $H''$, more formally:
$$[H' \Box H'']_2 =[H']_2 \Box [H'']_2.$$
\end{prop}
This observation suggested the definition of the Cartesian product of
$L2$-sections by constructing an appropriate labeling function for 
the product:
\begin{defn}[The Cartesian Product of L2-sections 
\cite{Bretto09:HyperCartProd, BrettoSilvestre10:Factorization}]
Let $\Gamma_i=(V_i,E'_i,\mathcal{L}_i)$ be the $L2$-section of the 
hypergraphs $H_i=(V_i,E_i)$, $i=1,2$. The Cartesian product of the 
$L2$-sections $\Gamma_1\Box \Gamma_2 = (V,E',\mathcal{L})$ consists of the 
graph $(V,E')=(V_1,E'_1)\Box(V_2,E'_2)$ and a labeling function 
$$\mathcal{L}=\mathcal{L}_1\Box \mathcal{L}_2: 
E'\rightarrow \mathbb{P}(E(H_1\Box H_2))$$
with 
\[\mathcal{L}\left(\{(x_1,y_1),(x_2,y_2)\}\right) = 
\begin{cases}
  \left\{\{x_1\}\times e\mid e\in\mathcal{L}_2\left(\{y_1,y_2\}\right)\right\},
  &\text{if }x_1=x_2 \\ 
  \left\{e\times \{y_1\}\mid e\in\mathcal{L}_1\left(\{x_1,x_2\}\right)\right\},
  &\text{if }y_1=y_2
\end{cases}\]
\end{defn}

\begin{lem}[\cite{Bretto09:HyperCartProd, BrettoSilvestre10:Factorization}]
  For all hypergraphs $H,H'$ we have:
  \begin{enumerate}
  \item $[H\Box H']_{L_2}=[H]_{L2}\Box[H']_{L2}$
  \item $[[H]_{L2}\Box[H']_{L2}]_{L2}^{-1} = 
    [[H]_{L2}]_{L2}^{-1}\Box[[H']_{L2}]_{L2}^{-1}$
  \end{enumerate}
\end{lem}

\begin{lem}[Distance Formula]
  \label{lem:distCart}
  For all hypergraphs $H,H'$ we have:
  $$d_{H\Box H'}((x,a),(y,b)) =   d_{H}(x,y) + d_{H'}(a,b)$$
\end{lem}
\begin{proof}
  Combining the results of Lemma \ref{lem:distH2}, Lemma
  \ref{prop:CartProd_2-sect} and the well-known Distance Formula for the
  Cartesian graph product (Corollary 5.2 in \cite{Hammack:2011a}) yields
  to the result.
\end{proof}

\subsection{Prime Factor Decomposition}

\begin{thm}[UPFD \cite{Imrich67:Mengensysteme}]
	\label{thm:UPFDCart}
Every connected hypergraph has a unique prime
factor decomposition w.r.t.\ the Cartesian product.
\end{thm}
The PFD of disconnected hypergraphs is in general not unique
\cite{IMKL-00,Hammack:2011a}.  Theorem~\ref{thm:UPFDCart} was also obtained
in \cite{OstHellmStad11:CartProd} using a different approach that
generalizes this result to infinite and directed hypergraphs, see Part
\ref{part:further aspects}.

Imrich and Peterin devised an algorithm for computing the PFD of connected
graphs $(V,E)$ in $O(|E|)$ time and space \cite{Imrich07:linear}.  Bretto
and Silvestre adapted this algorithm for the recognition of Cartesian
products of hypergraphs \cite{BrettoSilvestre10:Factorization}.  To this
end, the $L2$-sections of hypergraphs are used. We give here a short
outline of this algorithm.  For a given a connected hypergraph $H$ its
$L2$-section $[H]_{L2}$ is computed. Using the algorithm of Imrich and
Peterin one gets the PFD of $[H]_{2}$. This results in an edge coloring of
$[H]_{2}$, i.e., edges are colored with respect to the copies of the
corresponding prime factors.  After this, on has to check if the factors of
$[H]_{2}$ are the labeled prime factors of $[H]_{L2}$ and has to merge
factors if necessary. Finally, using the inverse $L2$-sections the prime
factors of $H$ are built back.  Although the PFD of the $2$-section
$[H]_{2} = (V,E')$ can be computed in $O(r(H)^2 \cdot |E| ) = O(|E'|)$
time, the check-and-merging-process together with the build back part for
the PFD of $H$ is more time-consuming and one ends in an overall time
complexity of $O((\log_2 |V|)^2 \cdot r(H)^3 \cdot|E|\cdot \Delta(H)^2)$
for a given hypergraph $H$.
  The PFD of connected simple hypergraphs $H=(V,E)$ with fixed maximum
  degree and fixed rank can then be computed in $O(|V||E|)$ time \cite{BrettoSilvestre10:Factorization}. 
The currently fastest algorithm is due to Hellmuth and Lehner \cite{HL:16}. 
In distinction from the method of Bretto et
al.\ this algorithm is in a sense conceptually simpler, as \emph{(1)} it is not
needed to transform the hypergraph $H$ into its so-called L2-section and
back and \emph{(2)} the test which (collections) of the putative factors 
are prime factors of $H$ follows a complete new idea based on
increments of fixed vertex-coordinate positions,
that allows an easy and efficient check to determine the PFD of $H$.

\begin{thm}[\cite{HL:16}]
	The PFD w.r.t.\ the Cartesian product of a hypergraph  $H = (V,E)$ with rank $r$ can be computed in 
	 $O(r^2|V||E|)$ time. 
	If we assume that $H$ has
	bounded rank, then this time-complexity can be reduced to $O(|E|\log^2(|V|))$. 
\label{thm:UPFDCartAlg}
\end{thm}

\subsection{Invariants}

Much of the literature on Cartesian hypergraph products is concerned with
relationships of invariants of the factors with those of the product. In
this section we compile the most salient results.

\begin{thm}[Automorphism Group \cite{Imrich67:Mengensysteme}]
	\label{thm:AutCart}
  The automorphism group of the Cartesian product of connected prime
  hypergraphs is isomorphic to the automorphism group of the disjoint union
  of the factors.
\end{thm}

\begin{thm}[$k$-fold Covering \cite{Doerfler79:CoversDirectProd}]
  \label{thm:k-covCart}
  Let $H_i'=(V_i',E_i')$ be a $k_i$-fold covering of the hypergraph
  $H_i=(V_i,E_i)$ via a covering projection $\pi_i$, $i=1,2$.  Then $H'_1
  \Box H'_2$ is a $k_1k_2$-fold cover of $H_1 \Box H_2$ via a covering
  projection $\pi$ induced naturally by $\pi_1$ and $\pi_2$, i.e, define
  $\pi$ by:
\begin{itemize}
\item[] $\pi((x,y)) = (\pi_1(x),\pi_2(y))$, for $(x,y) \in V'_1\times V'_2$, 
\item[] $\pi(\{x\}\times e_2) = \{\pi_1(x)\}\times \pi_2(e_2)$, for $x \in
  V'_1,\, e_2\in E'_2$,
\item[] $\pi(e_1\times\{y\}) = \pi_1(e_1)\times\{\pi_2(y)\}$, for $e_1 \in
  E'_1,\, y\in V'_2$.
\end{itemize}
\end{thm}

\begin{thm}[Conformal Hypergraphs \cite{Bretto06:Helly}]
  \label{thm:ConfCart}
  $H = H_1 \Box H_2$ is conformal if and only if $H_1$ and $H_2$ are
  conformal.
\end{thm}

\begin{thm}[Helly Property \cite{Bretto06:Helly}]
	\label{thm:HellyCart}
  $H = H_1 \Box H_2$ has the Helly property if and only if $H_1$ and $H_2$
  have the Helly property.
\end{thm}

\begin{thm}[Colored Hyperedge Property \cite{Bretto09:HyperCartProd}]
	\label{thm:ColorICart}
  If $H_1$ and $H_2$ have the colored hyperedge property then $H = H_1 \Box
  H_2$ has the colored hyperedge property.
\end{thm}

\begin{thm}[(Strong) Chromatic Number \cite{Bretto09:HyperCartProd}]
  \label{thm:ColorIICart}
  Let $\chi_i$ and $\chi$ (respectively $\gamma_i$ and $\gamma$) be the
  chromatic (resp. strong chromatic number) of $H_i$ and $H =
  \Box_{i=1}^n$. Then $\chi = \max_i\{\chi_i\}$ and $\gamma =
  \max_i\{\gamma_i\}$.
\end{thm}

\begin{thm}[Hamiltonicity I \cite{Sonntag89:HamCart}]
	\label{thm:HamICart}
  Let $H_1=(V_1,E_1)$ and $H_2=(V_2,E_2)$ be two hypergraphs that contain a
  Hamiltonian path.  Then $H_1\Box H_2$ contains an Hamiltonian cycle if
  and only if $|V_1||V_2|$ is even or at least one of $H_1$ or $H_2$ is not
  a bipartite graph.
\end{thm}

\begin{thm}[Hamiltonicity II \cite{Sonntag89:HamCart}]
	\label{thm:HamIICart}
  Let $H_1$ be a hypergraph with $n_1\geq 4$ vertices containing an
  Hamiltonian cycle and $H_2)$ be a hypergraph with that contains a
  Hamiltonian $p$-path with $p\leq 2\lfloor n_1/4 \rfloor$. Then $H_1\Box
  H_2$ contains a Hamiltonian cycle.
\end{thm}

\part{Direct Products}
\label{sec:DirectProducts}

In contrast to the Cartesian product, there are several different
possibilities to construct a direct product. We will consider four
constructions in detail: The direct product $\timesR$, which is closed
under the restriction on $r$-uniform hypergraphs, the direct product
$\weithutzl{\times}$, which preserves the minimal rank of the factors, the
direct product $\weitdacherl{\times}$, which preserves the maximal rank of the
factors and the direct product $\widetilde{\times}$, which does not preserve
any rank of its factors.

An alternative product, which we prefer to call the \emph{square product}
following the work of Ne{\v{s}}et{\v{r}}il and R{\"{o}}dl
\cite{NesestrilRoedl83:SquareProduct}, is also often called ``direct
product'' in the literature. It will be discussed in detail in
section~\ref{sect:square}. 

\section{The Direct Product for $\boldsymbol{r}$-uniform Hypergraphs}

An early construction of a direct hypergraph product
\cite{Doerfler80:CategoriesHpg} was motivated by the investigation of a
category of hypergraphs. The following product is categorical in the
category of $r$-uniform hypergraphs and is only defined for $r$-uniform
hypergraphs.

\subsection{Definition and Basic Properties}
\label{sec:rUniformityPreservingDirectProduct}

\begin{defn}[$r$-uniform direct product]
For two $r$-uniform hypergraphs $H_1=(V_1, E_1)$ and $H_2=(V_2,E_2)$ 
their direct product
$H_1\timesR H_2$ has vertex set $V_1\times V_2$ and the edge set
\begin{equation}
\begin{split}	\label{eq:Direct_r-uniform}
    E(H_1\timesR H_2):
    =\left\{e\in\binom{e_1\times e_2}{r}\mid e_i\in E_i  
      \text{\ and\ } p_i(e)\in E_i,\ i=1,2 \right\}.
\end{split}
\end{equation}
\end{defn}

This product is the restriction of minimal and maximal rank preserving
products to $r$-uniform hypergraphs, defined in the following two sections.
Most of the properties of the two products can indeed be inferred from the
corresponding results for the $\timesR$-product.

\subsection{Relationships with Graph Products}

For $r=2$, $\timesR$ is the direct graph product in the simple graph.
However, since it is only defined on $r$-uniform hypergraphs, it cannot be
generalized on the class of graphs with loops. 
Since $\timesR$ coincides with the minimal rank preserving
product $\weithutzl{\times}$ on $r$-uniform hypergraphs all results concerning $2$-sections
and $L2$-sections can be inferred from the respective results of
$\weithutzl{\times}$. 
In general there is no unit
element for $r$-uniform hypergraphs, hence the term \emph{prime} cannot be
defined for this product. As far as we know, nothing is known about the 
behavior of hypergraph invariants under this product. 

\section{The Minimal Rank Preserving Direct Product}
\label{sec:MinimalRankPreservingDirectProduct}

If one considers the direct product $H_1\timesR H_2$ of two $r$-uniform 
hypergraphs $H_1$ and $H_2$, one observes that an edge in $E(H_1\timesR H_2)$
satisfies the following two properties:
\begin{itemize}
\item[(E1)] All vertices of an edge differ in each coordinate.
\item[(E2)] The projection of an edge is an edge in the respective factor.
\end{itemize}
If one tries to generalize the product $\timesR$ to arbitrary non-uniform
hypergraphs, one always encounters edges in the corresponding hypergraph
product that cannot satisfy both (E1) and (E2).  Hence, a natural question
is how to extend the direct product $\timesR$ to a product of two
arbitrary, non-uniform hypergraphs in such a way that it satisfies at least
one the properties.

If we insist on Property (E1) we enforce an additional constraint, that is, 
the projections of an edge of the product hypergraph into the factors 
is an edge in at least one factor and subsets of edges in the other factors.
>From that point, we observe that the rank of the hypergraph product
equals the minimal rank of the factors. 

\subsection{Definition and Basic Properties}

This product was first defined by Sonntag in \cite{Sonntag:thesis} using 
the term ``Cartesian Product''.

\begin{defn}[Minimal Rank Preserving Direct Product \cite{Sonntag:thesis}]
  For two hypergraphs $H_1=(V_1,E_1)$ and $H_2=(V_2,E_2)$ their direct
  product $H_1\weithutzl{\times} H_2$ has vertex set $V_1\times V_2$.
  For two edges $e_1\in E_1$ and $e_2\in E_2$ let 
  ${r}^{-}_{e_1,e_2}=\min\{|e_1|,|e_2|\}$.
  The edge set is the defined as:		
\begin{equation}
\begin{split}	
    E(H_1\weithutzl{\times} H_2):
    =\left\{e\in\binom{e_1\times e_2}{r^{-}_{e_1,e_2}}\mid e_i\in E_i
		\text{\ and\ }  |p_i(e)|=r^{-}_{e_1,e_2},\ i=1,2 \right\}.
\end{split}
\end{equation}
\end{defn}

In other words, a subset $e=\{(x_1,y_1),\ldots,(x_r,y_r)\}$ of $V_1\times
V_2$ is an edge in $H_1\weithutzl{\times} H_2$ if and only if
\begin{itemize}
\item[(i)]$\{x_1,\ldots,x_r\}$ is an edge in $H_1$ and $\{y_1,\ldots,y_r\}$
  is the subset of an edge in $H_2$, or 
\item[(ii)]$\{x_1,\ldots,x_r\}$ is the subset of an edge in $H_1$ and
  $\{y_1,\ldots,y_r\}$ is an edge in $H_2$
\end{itemize}

We have $p_1(e)=e_1$ and $p_2(e)=e_2$ provided $|e_1|=|e_2|$. If
$|e_i|<|e_j|$, then $p_i(e)= e_i$ and $p_j(e)\subset e_j$, $i,j\in\{1,2\}$.
Thus, the projections need not to be (weak) homomorphisms in general, but
they preserve adjacency, i.e., two vertices in a direct product
$\weithutzl{\times}$ hypergraph are adjacent, whenever they are adjacent in
both of the factors.

The direct product $\weithutzl{\times}$ is associative, it is commutative as an
immediate consequence of the symmetry of the definition.  Simple
set-theoretic considerations show that direct product $\weithutzl{\times}$ is
left and right distributive with respect to the disjoint union.  The direct
product $\weithutzl{\times}$ of two connected hypergraphs is not necessarily
connected, as one can observe for the simple case $K_2 \weithutzl{\times} K_2$.
For the rank and anti-rank, respectively, of the hypergraph $H=
H_1\weithutzl{\times} H_2$ holds:
\begin{align*}
  r(H_1\weithutzl{\times} H_2)&=\min\{r(H_1),r(H_2)\}\\
  s(H_1\weithutzl{\times} H_2)&=\min\{s(H_1),s(H_2)\}.
\end{align*}

\begin{lem} \label{lem:simple_direct_product1}
  The direct product $H_1\weithutzl{\times}H_2$ of
  simple hypergraphs is simple.
\end{lem}
\begin{proof}
  Let $H_1=(V_1,E_1)$ and $H_2=(V_2,E_2)$ be two simple
  hypergraphs.  Hence $s(H_i)\geq 2$ for $i=1,2$, and therefore
  $s(H_1\weithutzl{\times} H_2)=\min\{s(H_1),s(H_2)\}\geq 2$, which implies
  that $E(H_1\weithutzl{\times} H_2)$ contains no loops.
  Assume that there is an edge $e$ contained in an edge $e'$, where
  $e\in \binom{e_1\times e_2}{r^{-}_{e_1,e_2}}$ and 
  $e'\in \binom{e'_1\times e'_2}{r^{-}_{e'_1,e'_2}}$
  $e_i,e'_i\in E_i$.
  Notice, that $p_i(e)\subseteq p_i(e')$, $i=1,2$ must hold.
  W.l.o.g.\ suppose $|e_1|\leq |e_2|$, which implies $p_1(e)=e_1$.  It
  follows $e_1=p_1(e)\subseteq p_1(e')\subseteq e'_1$, and since $H_1$ is
  simple, $e_1=e'_1$ must hold and hence $p_1(e')=e'_1$.  From this we can
  conclude $|e|=|e_1|=|e'_1|=|p_1(e')|=|e'|$ and therefore $e=e'$.
\end{proof}

The direct product $\weithutzl{\times}$ does not have a unit, both in the class
of simple and non-simple hypergraphs, i.e., neither the one vertex
hypergraph $K_1$, nor the vertex with a loop, $\mathscr{L}K_1$, is a unit
for the direct product $\weithutzl{\times}$.

\subsection{Relationships with Graph Products}

The restriction of the direct product $\weithutzl{\times}$ to $r$-uniform
hypergraphs is the product $\timesR$ defined in Equation
\eqref{eq:Direct_r-uniform}, hence the restriction of $\weithutzl{\times}$
simple graphs coincides with the direct graph product.
But since the direct product $\weithutzl{\times}$ has no unit, we can conclude
that it does not coincide with the direct graph product in the class of
graphs with loops.

\begin{lem} \label{lem:DirectProd_2-sect} 
  The $2$-section of
  the direct product $H= H' \weithutzl{\times} H''$ is the direct product of
  the $2$-section of $H'$ and the $2$-section of $H''$, more formally: 
  $$[H' \weithutzl{\times} H'']_2 =[H']_2 \times [H'']_2.$$
\end{lem}

\begin{proof}
  Let $p_1$ and $p_2$ denote $p_{H'}$ and $p_{H''}$, respectively.
  By definition of the $2$-section and the direct product, $[H]_2$ and
  $[H']_2 \times [H'']_2$ have the same vertex set.
  Thus we need to show that the identity mapping 
  $V([H]_2)\rightarrow V([H']_2 \times [H'']_2)$ is an isomorphism.
  We have:
  $\{x,y\}\in E([H]_2)\Leftrightarrow 
  \exists\, e\in E(H): \{x,y\}\subseteq e,\ x\neq y
  \Leftrightarrow \exists e'\in E(H'), e''\in E(H''): 
  \{p_{1}(x),p_{1}(y)\}\subseteq p_{1}(e)\subseteq e',\ 
  p_{1}(x) \neq p_{1}(y) \text{ and } 
  \{p_{2}(x),p_{2}(y)\}\subseteq p_{2}(e)\subseteq e'',\  
  p_{2}(x)\neq p_{2}(y)
  \Leftrightarrow \{p_{1}(x),p_{1}(y)\}\in E([H']_2), 
  \{p_{2}(x),p_{2}(y)\}\in E([H'']_2)
  \Leftrightarrow \{x,y\}\in E([H']_2\times [H'']_2)$.
\end{proof}

\begin{defn}[The Direct Product of L2-sections]
  Let $\Gamma_i=(V_i,E'_i,\mathcal{L}_i)$ be the $L2$-section of the
  hypergraphs $H_i=(V_i,E_i)$, $i=1,2$. The direct product of the
  $L2$-sections $\Gamma_1\weithutzl{\times} \Gamma_2 = (V,E',\mathcal{L})$
  consists of the graph $(V,E')=(V_1,E'_1)\times(V_2,E'_2)$ and a labeling
  function
  $$\mathcal{L}=\mathcal{L}_1\weithutzl{\times} 
  \mathcal{L}_2: E'\rightarrow \mathbb{P}(E(H_1\weithutzl\times H_2))$$
  assigning to each edge $e' = \{(x_1,y_1),(x_2,y_2)\} \in E'$ a label
  \begin{align*}
    \mathcal{L}\left(e'\right) = 
    \{ 
    e' \cup A \mid  
    A\in \mathcal{A}(e',e_1,e_2), e_i\in \mathcal{L}_i(e'_i),|p_i(A)| = 
    r^{-}_{e_1,e_2}-2, i=1,2
    \}		    
  \end{align*}
  with $e'_1 =\{x_1, x_2\}$, $e'_2=\{y_1,y_2\}$ and
  $$\mathcal{A}(e',e_1,e_2) =
  \binom{(e_1\setminus e'_1) \times (e_2\setminus e'_2)}{r^{-}_{e_1,e_2}-2}. $$
\end{defn}

A short direct computation shows that 
$$[H\weithutzl\times H']_{L_2}=[H]_{L2}\weithutzl\times[H']_{L2}$$ 
holds for all simple hypergraphs $H,H'$.

\begin{lem}[Distance Formula]
  Let $(x_1,x_2)$ and $(y_1,y_2)$ be two vertices of the direct product
  $H = H_1\weithutzl{\times} H_2$. Then 
  $$d_{H}((x_1,x_2),(y_1,y_2)) =   
  \min\{n\in \mathbb{N} \mid\text{each factor } H_i 
  \text{ has a walk } P_{x_i,y_i}
  \text{ of length } n\},$$
  where it is understood that 			
  $d_{H}((x_1,x_2),(y_1,y_2)) = \infty$
  if no such $n$ exists.
  \label{lem:distMinRankDir} 
\end{lem}
\begin{proof}
  Combining the results of Lemma \ref{lem:distH2}, Lemma
  \ref{lem:DirectProd_2-sect} and the Distance
  Formula for the direct graph product (Proposition 5.8 in \cite{Hammack:2011a})
  yields to the result.
\end{proof}

\begin{cor}
  Let $(x_1,x_2)$ and $(y_1,y_2)$	be two vertices of the direct product
  $H=H_1\timesR H_2$. Then 
  $$d_{H}((x_1,x_2),(y_1,y_2)) =   
  \min\{n\in \mathbb{N} \mid \text{ each factor } H_i 
  \text{ has a walk } P_{x_i,y_i}
  \text{ of length } n\},$$
  where it is understood that 			
  $d_{H}((x_1,x_2),(y_1,y_2)) = \infty$
  if no such $n$ exists.
  \label{cor:distUniformDir} 
\end{cor}

The absence of a unit implies that there is no meaningful PFD. To the authors'
knowledge, hypergraphs invariants have not been studied for the product 
$\weithutzl{\times}$.

\section{The Maximal Rank Preserving Direct Product}
\label{sec:MaximalRankPreservingDirectProduct}

As discussed in Section \ref{sec:MinimalRankPreservingDirectProduct}, 
if one tries to generalize the product $\timesR$ to arbitrary non-uniform
hypergraphs, one always encounters edges in the corresponding hypergraph
product that cannot satisfy both (E1) and (E2):

\begin{itemize}
\item[(E1)] All vertices of an edge differ in each coordinate.
\item[(E2)] The projection of an edge is an edge in the respective factor.
\end{itemize}

In order to extend the direct product $\timesR$ to a product of two
arbitrary, non-uniform hypergraphs in such a way that it satisfies at least
one of the properties, we insist now on condition (E2) and  
claim that the size of an edge in the
product hypergraph coincides with the size of at least one of its
projections and thus, the rank of the product hypergraph is exactly the
maximal rank of one of its factors.  

\subsection{Definition and Basic Properties}

\begin{defn}[Maximal Rank Preserving Direct Product]
  For two hypergraphs $H_1=(V_1,E_1)$ and $H_2=(V_2,E_2)$ their direct
  product $H_1\weitdacherl{\times} H_2$ has vertex set $V_1\times V_2$.  For two
  edges $e_1\in E_1$ and $e_2\in E_2$ let
  ${r}^{+}_{e_1,e_2}=\max\{|e_1|,|e_2|\}$.  The edge set is the defined as:
\begin{equation}
\begin{split}	
  E(H_1\weitdacherl{\times} H_2)
  := \left\{e\in\binom{e_1\times e_2}{r^{+}_{e_1,e_2}}\mid e_i\in E_i
    \text{\ and\ }  p_i(e)=e_i,\ i=1,2 \right\}.
\end{split}
\end{equation}
\end{defn}

In other words, a subset $e=\{(x_1,y_1),\ldots,(x_r,y_r)\}$ of $V_1\times
V_2$ is an edge in $H_1\weitdacherl{\times} H_2$ if and only if
\begin{itemize}
\item[(i)] $\{x_1,\ldots,x_r\}$ is an edge in $H_1$ 
  and there is an edge $f\in E_2$ of $H_2$ 
  such that $\{y_1,\ldots y_r\}$ is a multiset of elements of $f$,
  and $f\subseteq\{y_1,\ldots,y_r\}$, or
\item[(ii)] $\{y_1,\ldots,y_r\}$ is an edge in $H_2$ 
  and there is an edge $e\in E_1$ of $H_1$ 
  such that $\{x_1,\ldots x_r\}$ is a multiset of elements of $e$,
  and $e\subseteq\{x_1,\ldots,x_r\}$.
\end{itemize}	
For $e\in E(H_1\weitdacherl{\times} H_2)$  holds: if $|e|=|p_i(e)|$, 
$i\in\{1,2\}$, then
$p_i(x)\neq p_i(y)$ for all $x,y\in e$ with $x\neq y$. 

The direct product $\weitdacherl{\times}$ is associative, commutative, and
distributive with respect to the disjoint union. Contrary to the direct
product $\weithutzl{\times}$, projections of a product hypergraph into the
factors are, by definition, homomorphisms, i.e., projections of hyperedges are
hyperedges in the respective factors.

The one vertex hypergraph with a loop $\mathscr{L}K_1$ is a unit for the
direct product $\weitdacherl{\times}$ in the class of hypergraphs with
loops.  In the class of simple hypergraphs, this product has no unit.  The
direct product $\weitdacherl{\times}$ of two connected hypergraphs is not
necessarily connected, since it need not to be connected in the class of
graphs.  For the rank and anti-rank, respectively, of the hypergraph
$H_1\weitdacherl{\times}H_2$ holds:
\begin{align*}
  r(H_1\weitdacherl{\times} H_2)&=\max\{r(H_1),r(H_2)\}\\
  s(H_1\weitdacherl{\times} H_2)&=\max\{s(H_1),s(H_2)\}.
\end{align*}

\begin{lem}	\label{lem:SimpleDirect2}
  The direct product $H_1\weitdacherl{\times}H_2$ of
  simple hypergraphs is simple.
\end{lem}
\begin{proof}
  Therefore, let $H_1=(V_1,E_1)$ and $H_2=(V_2,E_2)$ be two simple
  hypergraphs.  Hence $s(H_i)\geq 2$ for $i=1,2$, and therefore it holds
  $s(H_1\weitdacherl{\times} H_2)=\max\{s(H_1),s(H_2)\}\geq 2$, 
  which implies that
  $E(H_1\weitdacherl{\times} H_2)$ contains no loops.  
  Assume that there is an edge
  $e$ contained in an edge $e'$, where $e\in \binom{e_1 \times
    e_2}{r^+_{e_1,e_2}}$ and $e'\in\binom{e'_1 \times e'_2}{r^+_{e'_1,e'_2}}$,
  $e_i,e'_i\in E_i$.  Notice, that $p_i(e)\subseteq p_i(e')$, $i=1,2$ must
  hold.  Hence, $e_i=p_i(e)\subseteq p_i(e')=e'_i$, which implies
  $e_i=e_i'$, since $H_1$ and $H_2$ are simple.  It follows $|e|=|e'|$ and
  therefore $e=e'$.  Thus, $H_1\weitdacherl{\times}H_2$ is simple.
\end{proof}

\subsection{Relationships with Graph Products}

If we restrict the direct product $\weitdacherl{\times}$ to $r$-uniform
hypergraphs, we recover the product $\timesR$ defined by Equation
\eqref{eq:Direct_r-uniform}.  In particular, this product coincides with
the direct graph product in the class of simple graphs. Moreover, the
restriction of this product to graphs coincides with the direct graph
product in general, also in the class of not necessarily simple graphs with
loops.

In contrast to the direct product $\weithutzl{\times}$, however, the
$2$-section of the direct product of two arbitrary hypergraphs is not the
direct graph product of the $2$-sections of the hypergraphs, except in the
special case of $r$-uniform hypergraphs. To see this, consider as an
example the product $K_2 \weitdacherl{\times} (V,\{V\})$ with $|V|=3$.

Despite its appealing properties, the product $\weitdacherl{\times}$ has
not been studied in the literature. It is unknown, in particular, under
which conditions it admits a unique PFD. The prime factor theorems for the
direct product need non-trivial preconditions even in the class of graphs.
We do not expect that it will be a particularly simple problem to establish
a general UPFD theorem.

\section{A Direct Product that does not preserve Rank}

For the sake of completeness, and as an example for the degrees of freedom
inherent in the definition of hypergraph products that generalize graph
products, we consider the product $\widetilde{\times}$. Its restriction to
$2$-uniform hypergraphs coincides with the direct graph product. However,
it does not preserve $r$-uniformity in general.  For brevity we omit
proofs, which can be found in \cite{Gringmann:thesis}, throughout this
section.

\subsection{Definition and Basic Properties}

\begin{defn}[non-rank-preserving Direct Product]
  For two hypergraphs $H_1=(V_1,E_1)$ and $H_2=(V_2,E_2)$, 
  we define their \emph{direct product $\widetilde{\times}$} by
  the edge set 
  \begin{align*}
    E(H_1\widetilde{\times} H_2):=
    \big\{\{(x,y)\}\cup \big((e\setminus\{x\})\times(f\setminus\{y\}) \big)\mid 
    x\in e\in  E_1;\ y\in f\in E_2\big\}	
  \end{align*}
\end{defn}

The projections $p_1$ and $p_2$ of a product hypergraph
$H=H_1\widetilde{\times}H_2$ into its factors $H_1$ and $H_2$,
respectively, are homomorphisms.
		
It is not hard to verify that the direct product $\widetilde{\times}$ is
associative, commutative, and both left and right distributive together
with the disjoint union as addition.  The direct product $H_1\widetilde{\times}
H_2$ of simple hypergraphs $H_1,H_2$ is simple.  The direct product
$\widetilde{\times}$ does not have a unit, neither in the class of simple
hypergraphs, nor in the class of non-simple hypergraphs.  The direct
product $\widetilde{\times}$ of two connected hypergraphs is not necessarily
connected, as one can observe for the simple case $K_2 \widetilde{\times} K_2$.
For the rank and anti-rank, respectively, of the hypergraph product
$H_1\widetilde{\times}H_2$ holds:
\begin{align*}
  r(H_1\widetilde{\times} H_2)=(r(H_1)-1)(r(H_2)-1)+1\\
  s(H_1\widetilde{\times} H_2)=(s(H_1)-1)(s(H_2)-1)+1
\end{align*}
In general, therefore, the (anti-)rank of a product will not be the
(anti-)rank of one of its factors.

\subsection{Relationships with Graph Products}

If we restrict the definition of this product to $2$-uniform hypergraphs,
i.e., simple graphs, we have: $e\subseteq V(G_1\widetilde{\times} G_2)$ is
an edge in $G_1\widetilde{\times} G_2$ iff $E=\{(x,y)(x',y')\}$ and
$\{x,x'\}$ is an edge in $G_1$ and $\{y,y'\}$ is an edge in $G_2$.  This is
exactly the definition of the direct graph product.

Similar to the direct product $\weitdacherl{\times}$, the $2$-section of
the direct product of two arbitrary hypergraphs is not the direct graph
product of the $2$-sections of the hypergraphs.  This can be easily
verified on the product $K_2 \widetilde{\times} (V,\{V\})$ with $|V|=3$.

Since the direct product $\widetilde{\times}$ has no unit, we can conclude that
it does not coincide with the direct graph product in the class of graphs
with loops.  The absence of a unit implies that there is no meaningful PFD.
To our knowledge, no further results are available on this product.

\part{Strong Products}

The strong product of graphs can be interpreted as a superposition of the
edges of the Cartesian and the direct graph products. Here we explore the
corresponding constructions for hypergraphs: Let the edge set of a strong
product $H=H_1\overset{*}{\boxtimes} H_2$, $*=\weitdacherl{ }$, 
$\weithutzl{ }$ of
two hypergraphs $H_1$ and $H_2$ be
\[ E(H_1\overset{*}{\boxtimes} H_2)= E(H_1\Box H_2)\cup
E(H_1\overset{*}{\times} H_2),\] 
where $ E(H_1\overset{*}{\times} H_2)$ is the edge set of one of the respective 
direct products discussed in the previous section.

\section{The Normal Product}
This particular strong product was first introduced by Sonntag
\cite{Sonntag90:NormalProd, Sonntag91:Corrigendum,
  Sonntag92:Hamiltonicity}.  Following the terminology of Sonntag we call
this product \emph{normal product} $\weithutzl{\boxtimes}$.

\subsection{Definition and Basic Properties}

A subset $e=\{(x_1,y_1),\ldots,(x_r,y_r)\}$ of $V_1\times V_2$ is an edge 
in $H_1\weithutzl{\boxtimes} H_2$ if and only if
\begin{itemize}
\item[(i)] $\{x_1,\ldots, x_r\}$ is an edge in $H_1$ and 
  $y_1=\ldots = y_r\in V(H_2)$, or
\item[(ii)] $\{y_1,\ldots,y_r\}$ 
  is an edge in $H_2$ and $x_1=\ldots = x_r\in V(H_1)$, or
\item[(iii)]$\{x_1,\ldots,x_r\}$ is an edge in $H_1$ and 
  $\{y_1,\ldots,y_r\}$ is the subset of an edge in $H_2$, or
\item[(iv)]$\{y_1,\ldots,y_r\}$ is an edge in $H_2$ and 
	$\{x_1,\ldots,x_r\}$ is the subset of an edge in $H_1$.
\end{itemize}
An edge $e$ that is of type $(i)$ or $(ii)$ is called \emph{Cartesian edge}
and it holds $e\in E(H_1\Box H_2)$, an edge $e$ of type $(iii)$ or $(iv)$
is called \emph{non-Cartesian edge} and it holds $e\in
E(H_1\weithutzl{\times}H_2)$.  Notice that $|e\cap e'|\leq 1$ holds for all
$e\in E(H_1\Box H_2)$ and $e'\in E(H_1\weithutzl{\times}H_2)$,
hence $E(H_1\Box H_2)\cap E(H_1\weithutzl{\times}H_2)=\emptyset$ if $H_1, H_2$
contain no loops.

For the same reasons as for the direct product $\weithutzl{\times}$, the
projections need not to be (weak) homomorphisms in general, but they
preserve adjacency or adjacent vertices are mapped into the same vertex.
The normal product $\weithutzl{\boxtimes}$ is associative, commutative, and
distributive w.r.t the disjoint union and has $K_1$ as unit element. 
Since $E(H_1\Box H_2)$ is a spanning partial hypergraph of
$E(H_1\weithutzl{\boxtimes} H_2)$, we can conclude  that the
normal product $H_1 \weithutzl{\boxtimes} H_2$ is connected if and only if 
$H_1$ and $H_2$ are connected hypergraphs.
For the rank and anti-rank, respectively, of a normal product hypergraph
$H_1\weithutzl{\boxtimes}  H_2$ holds:
\begin{align*}
  r(H_1\weithutzl{\boxtimes} H_2)=\max\{r(H_1),r(H_2)\}\\
  s(H_1\weithutzl{\boxtimes} H_2)=\min\{s(H_1),s(H_2)\}.
\end{align*}

\begin{lem} \label{lem:simple_strong_product1} The normal product $H_1
  \weithutzl{\boxtimes} H_2$ of simple hypergraphs $H_1, H_2$ is simple.
\end{lem}
\begin{proof}
  This follows immediately from Lemma~\ref{lem:simple_direct_product1}, the
  fact that the Cartesian product of simple hypergraphs is simple and that
  the intersection of a Cartesian and a non-Cartesian edge contains at most
  one vertex.
\end{proof}

\subsection{Relationships with Graph Products}

The restriction of the normal product $\weithutzl{\boxtimes}$ to
$2$-uniform hypergraphs coincides with the strong graph product.  But it
does not coincide with the strong graph product in the class of graphs with
loops since the direct product $\weithutzl{\times}$ does not coincide with
the direct graph product within this class.

In the class of graphs there is a well known relation between the direct
and the strong graph product. The strong product can be considered as a
special case of the direct product \cite{IMKL-00}: for a graph $G\in\Gamma$
let $\mathscr{L}G$ denote the graph in $\Gamma_0$, which is formed from $G$
by adding a loop to each vertex of $G$.  On the other hand, for a graph
$G'\in \Gamma_0$ let $\mathscr{N}G'$ denote the graph in $\Gamma$ which
emerges from $G'$ by deleting all loops.  Then we have for $G_1,G_2\in
\Gamma$:
\begin{equation}
  G_1\boxtimes G_2=\mathscr{N}(\mathscr{L}G_1\times\mathscr{L}G_2)
  \label{eq:SpecialCase}
\end{equation}	
This relationship, however, does not exist between the direct product
$\weithutzl{\times}$ and the normal product $\weithutzl{\boxtimes}$.

The next statement follows immediately from
Proposition~\ref{prop:CartProd_2-sect}, Lemma~\ref{lem:DirectProd_2-sect}
and the definition of the normal product $\weithutzl{\boxtimes}$:
\begin{lem}
  \label{lem:2sectionNormal}
  The $2$-section of
  the normal product $H= H' \weithutzl{\boxtimes} H''$ is the strong product
  of the $2$-section of $H'$ and the $2$-section of $H''$, more formally:
  $$[H' \weithutzl{\boxtimes} H'']_2 =[H']_2 \boxtimes [H'']_2.$$
\end{lem}

One can therefore define a meaningful normal product of L2-sections:

\begin{defn}[The Normal Product of L2-sections]
  Let $\Gamma_i=(V_i,E'_i,\mathcal{L}_i)$ be the $L2$-section of the
  hypergraphs $H_i=(V_i,E_i)$, $i=1,2.$ The normal product of the
  $L2$-sections $\Gamma_1\weithutzl{\boxtimes} \Gamma_2 = (V,E',\mathcal{L})$
  consists of the graph $(V,E')=(V_1,E'_1)\boxtimes(V_2,E'_2)$ and a
  labeling function   
  assigning to each edge $e' = \{(x_1,y_1),(x_2,y_2)\} \in E'$ a label
  $$\mathcal{L}=\mathcal{L}_1\weithutzl{\boxtimes} 
  \mathcal{L}_2: E'\rightarrow
  \mathbb{P}(E(H_1\weithutzl{\boxtimes} H_2))$$ with
  \[\mathcal{L}\left(\{(x_1,y_1),(x_2,y_2)\}\right) = 
\begin{cases}
  \left\{\{x_1\}\times e\mid e\in\mathcal{L}_2\left(\{y_1,y_2\}\right)\right\},
  &\text{if }x_1=x_2 \\ 
  \left\{e\times \{y_1\}\mid e\in\mathcal{L}_1\left(\{x_1,x_2\}\right)\right\},
  &\text{if }y_1=y_2 \\
      \mathcal{L}_1\weithutzl{\times} \mathcal{L}_2 
      \left(\{(x_1,y_1),(x_2,y_2)\}\right), 
  &\text{otherwise}.
  \end{cases}\]
\end{defn}

A straightforward but tedious computation shows 
$[H\weithutzl\boxtimes H']_{L_2}=[H]_{L2}\weithutzl\boxtimes[H']_{L2}$ 
for all simple hypergraphs $H,H'$.

\begin{lem}[Distance Formula]
  \label{lem:distNormal} 
  For all hypergraphs $H,H'$ we have:
  $$d_{H\weithutzl{\boxtimes} H'}((x,a),(y,b)) =   
  \max\{d_{H}(x,y), d_{H'}(a,b)\}$$
\end{lem}
\begin{proof}
  Combining the results of Lemma \ref{lem:distH2}, Lemma
  \ref{lem:2sectionNormal} and the well-known Distance
  Formula for the strong graph product (Corollary 5.5 in \cite{Hammack:2011a})
  yields to the result.
\end{proof}

\subsection{Invariants}

The normal product has not received much attention since its introduction
by Sonntag \cite{Sonntag90:NormalProd}.  In particular, it has not yet been
investigated regarding PFDs.

\begin{thm}[Hamiltonicity I \cite{Sonntag90:NormalProd}]	
  \label{thm:StrongProdHamI}
  Let $H_1=(V_1,E_1)$ and $H_2=(V_2,E_2)$ be two non-trivial hypergraphs
  s.t.  $H_1$ contains a Hamiltonian $p$-path, $p\in\mathbb N^+$, $H_2$
  contains a Hamiltonian $2$-path.  Suppose that one of the following
  conditions is satisfied
  \begin{itemize}
  \item[(1)] $|V_2|$ is even or $|V_1|=2$.
  \item[(2)] $|V_1|=3$ and $H_i$ is not isomorphic to $P_3$ or $H_j$ 
    contains no Hamiltonian $2$-path or 
    $|E_j|\neq\left\lfloor \frac{|V_j|}{2}\right\rfloor$, 
    $i,j\in\{1,2\}, i\neq j$.
  \item[(3)] $|V_1|\geq 4$ and there exists a Hamiltonian $2$-path 
    $P=(v_0,e_1,v_1,e_s,\ldots,e_{|V_2|-1},v_{|V_2|-1})$ in $H_2$, 
    s.t.\ there is an edge $e\in E_2$, and even indices 
    $i,i'$, $0\leq i<i'\leq |V_2|-1$ 
    with $\{v_i,v_{i'}\}\subseteq e$.
  \item[(4)] $|V_1|\geq 4$ and 
    $|V_2|\geq 2\left\lceil r/\left(\left\lfloor\frac{|V_1|-2}{p}+1\right
        \rfloor\right)\right\rceil-1$.
  \end{itemize}
  Then $H_1\weithutzl{\boxtimes} H_2$ contains a Hamiltonian cycle.	
\end{thm}

\begin{thm}[Hamiltonicity II \cite{Sonntag90:NormalProd}]
  \label{thm:StrongProdHamII}
  Let $H_1=(V_1,E_1)$ and $H_2=(V_2,E_2)$ be two non-trivial hypergraphs
  s.t.  $H_1$ contains a Hamiltonian $p$-path, $p\in\mathbb N^+$, $H_2$
  contains a Hamiltonian $2$-path and
  $|E_1|=\left\lfloor\frac{|V_1|-2}{p}\right\rfloor+1$.  Then
  $H_1\weithutzl{\boxtimes} H_2$ contains a Hamiltonian cycle if and only if at
  least one of the conditions $(1)-(4)$ of Theorem~\ref{thm:StrongProdHamI}
  is satisfied.
\end{thm}

\section{The Strong Product}

Given the ``nice'' properties of the direct product $\weitdacherl{\times}$,
the best candidate for a standard strong product of hypergraphs is
$\weitdacherl{\boxtimes}$ with edge set $E(H_1\Box H_2)\cup
E(H_1\weitdacherl{\times}H_2)$.

\subsection{Definition and Basic Properties}

A subset $e=\{(x_1,y_1),\ldots,(x_r,y_r)\}$ of $V_1\times V_2$ is an edge
in $H_1\weitdacherl{\boxtimes} H_2$ if and only if
\begin{itemize}
\item[(i)] $\{x_1,\ldots, x_r\} \in E(H_1)$ and 
  $y_1=\ldots = y_r\in V(H_2)$, or
\item[(ii)] $\{y_1,\ldots,y_r\}\in E(H_2)$ and $x_1=\ldots = x_r\in V(H_1)$, or
\item[(iii)] $\{x_1,\ldots,x_r\}\in E(H_1)$ and there is 
  an edge $f\in E(H_2)$ 
  such that $\{y_1,\ldots, y_r\}$ is a multiset of elements of $f$, 
  and $f\subseteq\{y_1,\ldots,y_r\}$, or
\item[(iv)] $\{y_1,\ldots,y_r\}\in E(H_2)$ and there is 
  an edge $f\in E(H_1)$ 
  such that $\{x_1,\ldots ,x_r\}$ is a multiset of elements of $f$, 
  and $f\subseteq\{x_1,\ldots,x_r\}$.
\end{itemize}

An edge $e$ that is of type $(i)$ or $(ii)$ is called \emph{Cartesian edge}
and it holds $e\in E(H_1\Box H_2)$, an edge $e$ of type $(iii)$ or $(iv)$
is called \emph{non-Cartesian edge} and it holds $e\in
E(H_1\weitdacherl{\times}H_2)$.

The strong product $\weitdacherl{\boxtimes}$ is associative, commutative, and
distributive w.r.t.\ the disjoint union and has $K_1$ as unit element.  The
projections into the factors are weak homomorphisms. 
Since $E(H_1\Box H_2)$ is a spanning partial hypergraph of
$E(H_1\weitdacherl{\boxtimes} H_2)$, we can conclude that strong product
$H_1\weitdacherl{\boxtimes} H_2$ is connected if and only if 
$H_1$ and $H_2$ are connected hypergraphs.
For the rank and anti-rank, respectively, of a strong product hypergraph
$H_1\weitdacherl{\boxtimes}  H_2$ holds:
\begin{align*}
  r(H_1\weitdacherl{\boxtimes} H_2)=\max\{r(H_1),r(H_2)\}\\
  s(H_1\weitdacherl{\boxtimes} H_2)=\min\{s(H_1),s(H_2)\}.
\end{align*}

\begin{lem} \label{lem:simple_strong_product2}
  The strong product $H_1\weitdacherl{\boxtimes}H_2$ 
  of simple hypergraphs $H_1$ and $H_2$ is simple.
\end{lem}
\begin{proof}
  Due to the fact that the Cartesian product and the direct product 
  $\weitdacherl{\times}$ of simple hypergraphs is simple, it remains to show, 
  that no Cartesian edge is contained in any non-Cartesian edge or vice versa.  
  Therefore, let $e$ be  a Cartesian edge and $e'$ a non-Cartesian 
  edge with $p_i(e')= e'_i$ for some $e'_i\in E_i$, $i=1,2$.
  Suppose first, $e'\subseteq e$.  Thus $|p_i(e)|=1$ for an $i\in\{1,2\}$
  and therefore $|p_i(e')|=1$, but $p_i(e')$ must be an edge in $H_i$.
  Hence, one of the factors would not be simple.
  Now let $e\subseteq e'$.  W.l.o.g.\ suppose $|p_1(e)|=1$, hence
  $p_2(e)=e_2$ for some $e_2\in E_2$.  Since $p_2(e)\subseteq p_2(e')=e'_2$
  and $H_2$ is simple, we can conclude $e_2=e'_2$.
  If $|e'|=|e'_1|$, then $p_1(x)\neq p_1(y)$ must hold for all $x,y\in e'$
  with $x\neq y'$, and therefore $|e|=1$ and hence $|e_2|=1$, which
  contradicts the fact that $H_2$ is simple.  If conversely $|e'|=|e'_2|$,
  we can conclude that $|e'|=|e|$, hence $e=e'$ and $|e'_1|=|p_1(e')|=1$
  which implies that $H_1$ is not simple, a contradiction.
\end{proof}

>From the arguments in the proof we can conclude that $E(H_1\Box H_2)\cap
E(H_1\weitdacherl{\times}H_2)=\emptyset$ holds if $H_1$ and $H_2$ are
simple. Moreover, one can show that $E(H_1\Box H_2)\cap
E(H_1\weitdacherl{\times}H_2)=\emptyset$ provided that $H_1$ and $H_2$ are
loopless.

\subsection{Relationships with Graph Products}

The restriction of the strong product $\weitdacherl{\boxtimes}$ to (not
necessarily simple) graphs (with or without loops) coincides with the
strong graph product. For simple hypergraphs $H_1$ and $H_2$ without loops,
furthermore, we have
\begin{equation}
  H_1 \weitdacherl{\boxtimes} H_2 = 
  \mathscr{N}(\mathscr{L}H_1\weitdacherl{\times}\mathscr{L}H_2)
  \label{eq:SpecialCase_Hpg}
\end{equation}	
Thus, the strong product can be considered as a special case of the
direct product, generalizing the well-known results about 
the direct and strong graph product, see Equ.(\ref{eq:SpecialCase}).

In contrast to the direct product $\weitdacherl{\times}$, the $2$-section
of the strong product $\weitdacherl{\boxtimes}$ coincides with the strong
graph product of the $2$-sections of its factors.
 
\begin{lem}  \label{lem:2sectionStrong}
  The $2$-section of the strong product 
  $H= H' \weitdacherl{\boxtimes} H''$ is the strong product 
  of the $2$-section of $H'$ and the $2$-section of $H''$, more formally:
  $$[H' \weitdacherl{\boxtimes} H'']_2 =[H']_2 \boxtimes [H'']_2.$$
\end{lem}

\begin{proof}
  Let $p_1$ and $p_2$ denote $p_{H'}$ and $p_{H''}$, respectively.  By
  definition of the $2$-section and the strong product, $[H]_2$ and $[H']_2
  \boxtimes [H'']_2$ have the same vertex set.  Thus we need to show that
  the identity mapping $V([H]_2)\rightarrow V([H']_2 \times [H'']_2)$ is an
  isomorphism.  We have: $\{x,y\}\in E([H]_2)\Leftrightarrow \exists\, e\in
  E(H): \{x,y\}\subseteq e,\ x\neq y$ $\Leftrightarrow $
  \begin{enumerate}
  \item $p_1(x) = p_1(y)$ and $(p_2(x),p_2(y)) \subseteq 
    p_2(e)\in E(H'')$, $p_2(x)\neq p_2(y)$ or
  \item $p_2(x) = p_2(y)$ and $(p_1(x),p_1(y)) \subseteq 
    p_1(e)\in E(H')$, $p_1(x)\neq p_1(y)$ or
  \item  $(p_1(x),p_1(y)) \subseteq p_1(e)\in E(H')$ and  
    $(p_2(x),p_2(y)) \subseteq p_2(e)\in E(H'')$, 
    $p_i(x)\neq p_i(y)$, $i=1,2$.
  \end{enumerate}
  $\Leftrightarrow $
  \begin{enumerate}
  \item $p_1(x) = p_1(y)$ and $(p_2(x),p_2(y)) \in E([H'']_2)$ or
  \item $p_2(x) = p_2(y)$ and $(p_1(x),p_1(y)) \in E([H']_2)$ or
  \item $(p_1(x),p_1(y)) \in E([H']_2)$ and $(p_2(x),p_2(y)) \in E([H'']_2)$.
  \end{enumerate}
  $\Leftrightarrow $
  $\{x,y\}\in E([H']_2\boxtimes [H'']_2)$.
\end{proof}

\begin{defn}[The Strong Product of L2-sections]
  Let $\Gamma_i=(V_i,E'_i,\mathcal{L}_i)$ be the $L2$-section of the 
  hypergraphs $H_i=(V_i,E_i)$, $i=1,2$.
  The strong product of the $L2$-sections 
  $\Gamma_1\weitdacherl{\boxtimes} \Gamma_2 = 
  (V,E',\mathcal{L})$ consists of the graph 
  $(V,E')=(V_1,E'_1)\boxtimes(V_2,E'_2)$ 
  and a labeling function $$\mathcal{L}=\mathcal{L}_1\weitdacherl{\boxtimes} 
  \mathcal{L}_2: E'\rightarrow \mathbb{P}(E(H_1\weitdacherl{\boxtimes} H_2))$$
  assigning to each edge $e' = \{(x_1,y_1),(x_2,y_2)\} \in E'$ a label
  $$\mathcal{L}\left(e'\right) = 
  (\mathcal{L}_1 \Box \mathcal{L}_2)(e') \cup \weitdacherl{\mathcal{L}}(e')$$
  where $(\mathcal{L}_1 \Box \mathcal{L}_2)(e') = \emptyset$ if 
  $x_1\neq x_2$ and $y_1\neq y_2$ and
  $$\weitdacherl{\mathcal{L}}(e') = 
  \begin{cases}
    \{e'\cup B\mid B\in\mathcal{B}(x_1,e'_2,e_2), 
    e_2\in \mathcal{L}_2(e'_2)\}, & 
    \text{ if } x_1=x_2\\
    \{e'\cup C\mid C\in\mathcal{C}(y_1,e'_1,e_1), 
    e_1\in \mathcal{L}_1(e'_1)\}, & 
    \text{ if } y_1=y_2 \\
    \{e'\cup D \mid D\in\mathcal{D}(e',e_1,e_2), e_1\in \mathcal{L}_1(e'_1), 
    e_2\in \mathcal{L}_2(e'_2)\}, &\text{ else}
  \end{cases}$$
  with $e'_1 =\{x_1, x_2\}$, $e'_2=\{y_1,y_2\}$ and 
  $$\mathcal{B}(x_1,e'_2,e_2)  = 
  \left\{B \in \binom{e_1\times (e_2\setminus e'_2)}{|e_2|-2}
    \mid x_1\in e_1 \in E_1, |e_1|<|e_2|, p_i(B\cup e')=
    e_i, i=1,2\right\} $$ and 
  $$\mathcal{C}(y_1,e'_1,e_1)  = 
  \left\{C \in \binom{ (e_1\setminus e'_1) \times e_2}{|e_1|-2}
    \mid y_1\in e_2 \in E_2, |e_2|<|e_1|, p_i(C\cup e')=e_i,
    i=1,2\right\} $$ 
  and 
  $$\mathcal{D}(e',e_1,e_2) = 
  \left\{ D \in \binom{(e_1\times e_2)\setminus 
      (e'_1\times e'_2)}{r^{+}_{e_1,e_2}-2} \mid
    p_i(D\cup e') = e_i, i=1,2 \right\}$$
\end{defn}

Again, one can show that 
$[H\weitdacherl\boxtimes H']_{L_2}=[H]_{L2}\weitdacherl\boxtimes[H']_{L2}$ 
holds for all simple hypergraphs $H,H'$.

\begin{lem}[Distance Formula]\label{lem:distStrong}
  For all hypergraphs $H,H'$ without loops we have:
  $$d_{H\weitdacherl{\boxtimes} H'}((x,a),(y,b)) =   
  \max\{d_{H}(x,y), d_{H'}(a,b)\}$$
\end{lem}
\begin{proof}
  Combining the results of Lemma \ref{lem:distH2}, Lemma
  \ref{lem:2sectionStrong} and the well-known Distance
  Formula for the strong graph product (Corollary 5.5 in \cite{Hammack:2011a})
  yields to the result.
\end{proof}

Although $\weitdacherl{\boxtimes}$ appears to be the most promising strong
product of hypergraphs, it has not been investigated in any detail so far.

\section{Alternative Constructions Generalizing the Strong Graph Product}

In order to generalize the strong graph product one can draw on an
abundance of ways to define strong hypergraph products.  To complete this
part, we suggest a few of possibilities that have not been considered in
the literature so far but might warrant further attention. 
	
\begin{enumerate}
\item $E(H_1\widetilde{\boxtimes}H_2)=E(H_1\Box H_2)\cup 
  E(H_1\widetilde{\times}H_2)$
\item $E(H_1\boxtimes H_2)=\left\{\binom{e\times e'}{r^-_{e,e'}}\mid e\in
    E_1, e'\in E_2\right\}$ 
\item $E(H_1\boxtimes H_2)=\left\{\binom{e\times e'}{r^+_{e,e'}}\mid e\in 
    E_1, e'\in E_2\right\}$ 
\item $E(H_1\boxtimes H_2)=\left\{\binom{e\times e'}{|e|}\mid e\in E_1, 
    e'\in E_2\right\}
  \cup \left\{\binom{e\times e'}{|e'|}\mid e\in E_1, e'\in E_2\right\}$
\item $E(H_1\boxtimes H_2)=\left\{\binom{e\times e'}{k}\mid e\in E_1, 
    e'\in E_2, r^-_{e,e'}\leq k\leq r^+_{e,e'} \right\}$
\end{enumerate}
	
\part{Lexicographic and Costrong Products}

\section{The Lexicographic Product}

The lexicographic product is the only non-commutative product treated in
this survey.  The lexicographic product of hypergraphs has received
considerable attention in the literature \cite{Sonntag91:HamTraceLexico,
  Sonntag92:Hamiltonicity, GasztImrich71:lexico, DoerflerImrich70:Xsum,
  Hahn81:lexico, Hahn80:DiHpg, BKL11:NetworkCoding,
  Doerfler78:DoubleCovers}.

\subsection{Definition and Basic Properties}

\begin{defn}[The Lexicographic Product 
  \cite{DoerflerImrich70:Xsum}]
  Let $H_1=(V_1,E_1)$ and $H_2=(V_2,E_2)$ be two hypergraphs.
  The \emph{lexicographic product} $H=H_1\circ H_2$ has vertex set
  $V(H)=V_1\times V_2$
  and edge set
  \[E(H)=\left\{e\subseteq V(H):  
    p_1(e)\in E_1, |p_1(e)| = |e|\right\}
  \cup\left\{\{x\}\times e_2\mid x\in V_1, e_2\in E_2 \right\}.\] 
\end{defn}

Since $|p_1(e)| = |e|$ there are $|e|$ vertices of $e$ that have pairwise
different first coordinates.  A related construction was also explored in
\cite{DoerflerImrich70:Xsum}:
\begin{defn}[$X$-join of hypergraphs]
  Let $X=(V(X),E(X))$ be a hypergraph and let $\{H(x)\mid x\in V(X)\}$ be a
  set of arbitrary pairwise disjoint hypergraphs, each of them associated with a vertex $x\in V(X)$.  
  The \emph{$X$-join} of $\{H(x)\mid
  x\in V(X)\}$ is the hypergraph $Z=(V(Z),E(Z))$ with
  \begin{align*}
    V(Z)=&\bigcup_{x\in V(X)}V(H(x))\quad\text{and}\\
    e\in E(Z)\Leftrightarrow &e\in E(H(x)),\quad\text{ or }\\
    &|e\cap V(H(x))|\leq 1 \text{ and }
    \{x\mid e\cap V(H(x))\neq\emptyset\}\in E(X).
  \end{align*}
\end{defn}
If $X\cong K_2$, then $Z$ is also called \emph{join} of $H_1$ and $H_2$,
$Z=H_1\oplus H_2$. The $X$-join is a generalization of the lexicographic
product in the following sense: If $Z$ is an $X$-join of hypergraphs
$\{H(x)\mid x\in V(X)\}$ and if $H(x)\cong H$ for all $x\in V(X)$ then
$Z$ is equivalent to the lexicographic product $X\circ H$.

The lexicographic product of two simple hypergraphs is simple.  It is 
associative, has the single vertex graph $K_1$ as a unit
element, and is right-distributive with respect to the join and the disjoint
union of hypergraphs.  Additionally, we have the following
left-distributive properties w.r.t.\ join and disjoint union:
\begin{align*}
  \overline{K}_n\circ (H_1+H_2) &= 
  \overline{K}_n\circ H_1 + \overline{K}_n\circ H_2\\
  K_n\circ (H_1\oplus H_2) &= K_n\circ H_1 \oplus K_n\circ H_2
\end{align*}
for all hypergraphs $H_1,H_2$ \cite{DoerflerImrich70:Xsum,
  GasztImrich71:lexico}.  The lexicographic product is not commutative in
general. 
\begin{thm}[\cite{DoerflerImrich70:Xsum}]
  Let $H_1$ and $H_2$ be two non-trivial connected finite hypergraphs.  Then
  $H_1\circ H_2\cong H_2\circ H_1$ only if $H_1$ and $H_2$ are complete
  graphs or $H_1$ and $H_2$ are both powers of some hypergraph $H$.
\end{thm}

Connectedness of the lexicographic product depends only on the first
factor. More precisely, $H_1\circ H_2$ is a connected hypergraph if and
only if $H_1$ is connected.  For the rank and anti-rank, respectively, of a
lexicographic product hypergraph $H_1\circ H_2$ holds:
\begin{align*}
	r(H_1\circ H_2) &=\max\{r(H_1),r(H_2)\}\\
	s(H_1\circ H_2) &=\min\{s(H_1),s(H_2)\}.
\end{align*}

The projection $p_1$ of a lexicographic product $H=H_1\circ H_2$ of two
hypergraphs $H_1$, $H_2$ into the first factor $H_1$ is a weak
homomorphism.  The $H_2$-layer through $w$, $H_2^w$, is the partial
hypergraph of $H$ induced by all vertices of $H$ which differ from a given
vertex $w\in V(H)$ exactly in the second coordinate,
and it holds:
\[
H_2^w=
\left\langle\left\{\left(p_1(w),t\right)\mid t\in
    V(H_2)\right\}\right\rangle 
\cong H_2.
\]

\subsection{Relationships with Graph Products}

The restriction of the lexicographic product on graphs coincides with the
usual lexicographic graph product.
\begin{lem}
  The $2$-section of $H= H' \circ H''$ is the lexicographic product 
  of the $2$-section of $H'$ and the $2$-section of $H''$, more formal:
  $$[H' \circ H'']_2 =[H']_2 \circ [H'']_2.$$
  \label{lem:2sectionLexi}
\end{lem}
\begin{proof}
  Let $p_1$ and $p_2$ denote $p_{H'}$ and $p_{H''}$, respectively.  By
  definition of the $2$-section and the lexicographic product, $[H]_2$ and
  $[H']_2 \boxtimes [H'']_2$ have the same vertex set.  Thus, we need to
  show that the identity mapping $V([H]_2)\rightarrow V([H']_2 \times
  [H'']_2)$ is an isomorphism.  We have: $\{x,y\}\in
  E([H]_2)\Leftrightarrow \exists\, e\in E(H): \{x,y\}\subseteq e,\ x\neq
  y$ $\Leftrightarrow $
  \begin{enumerate}
  \item $(p_1(x),p_1(y)) \subseteq p_1(e)\in E(H')$ such 
    that $p_1(x)\neq p_1(y)$ or
  \item $p_1(x) = p_1(y)$ and $(p_2(x),p_2(y)) \subseteq p_2(e)\in E(H'')$ 
    such that $p_2(x)\neq p_2(y)$.
  \end{enumerate}
  $\Leftrightarrow $
  \begin{enumerate}
  \item $(p_1(x),p_1(y)) \in E([H']_2)$ or
  \item $p_1(x) = p_1(y)$ and $(p_2(x),p_2(y)) \in E([H'']_2)$.
  \end{enumerate}
  $\Leftrightarrow $
  $\{x,y\}\in E([H']_2\circ [H'']_2)$.
\end{proof}

\begin{defn}[The Lexicographic Product of $L2$-sections]
  Let $\Gamma_i=(V_i,E'_i,\mathcal{L}_i)$ be the $L2$-section of the 
  hypergraphs $H_i=(V_i,E_i)$, $i=1,2$.
  The lexicographic product of the $L2$-sections 
  $\Gamma_1\circ \Gamma_2 = (V,E',\mathcal{L})$ consists of the 
  graph $(V,E')=(V_1,E'_1)\circ(V_2,E'_2)$ and a labeling function 
  $$\mathcal{L}=\mathcal{L}_1\circ \mathcal{L}_2: E'\rightarrow 
  \mathbb{P}(E(H_1\circ H_2))$$
  assigning to each edge $e' = \{(x_1,y_1),(x_2,y_2)\} \in E'$ a label
  \[\mathcal{L}\left(\{(x_1,y_1),(x_2,y_2)\}\right) = 
  \begin{cases}
    \left\{(e,f(e)) \mid  e\in \mathcal{L}_1(\{x_1,x_2\}), 
      f\in F(e,e')\right\},
    &\text{ if $x_1\neq x_2$}\\
    \left\{\{x_1\}\times e \mid e\in\mathcal{L}_2
      \left(\{y_1,y_2\}\right)\right\},
    &\text{ otherwise } 
  \end{cases}\]
  where $F(e,e') = \{f:e\rightarrow V_2 \mid \text f(x_1)=y_1, f(x_2)=y_2\}$
  and $(e,f(e))$ denotes the set $\{(x_1,f(x_1)),\ldots,(x_k,f(x_k))\}$ 
  for $e=\{x_1,\ldots,x_k\}$.
\end{defn}

A straightforward computation shows that
$[H\circ H']_{L_2}=[H]_{L2}\circ[H']_{L2}$ 
holds for all simple hypergraphs $H,H'$.

\begin{lem}[Distance Formula]
  Let $(x_1,x_2)$ and $(y_1,y_2)$ be two vertices of the lexicographic
  product $H = H_1 \circ H_2$. Then
  \begin{equation*}
    d_{H}((x_1,x_2),(y_1,y_2)) = 
    \begin{cases} d_{H_1}(x_1,y_1), & \text{if }x_1\neq y_1\\
      d_{H_2}(x_2,y_2), &\text{if }x_1 = y_1 \textrm{ and } 
      \deg_{H_1}(x_1) = 0 \\
      \min\{d_{H_2}(x_2,y_2), 2\}, &\text{if }x_1 = y_1 \textrm{ and }
      \deg_{H_1}(x_1) \neq 0\\
    \end{cases}
  \end{equation*}
  \label{lem:distLexi} 
\end{lem}
\begin{proof}
  Combining the results of Lemma \ref{lem:distH2}, Lemma
  \ref{lem:2sectionLexi} and the Distance
  Formula for the lexicographic graph product 
  (Proposition 5.12 in \cite{Hammack:2011a})
  yields to the result.
\end{proof}

\subsection{Prime Factor Decomposition}

Similar conditions as for graphs are known for the uniqueness of the PFD of
a hypergraph w.r.t.\ the lexicographic product:

\begin{lem}[\cite{GasztImrich71:lexico}]
  Let $H$ be a hypergraph without isolated vertices and $m, n$ natural
  numbers. Then $\overline{K}_n\circ H+\overline{K}_m$ is prime with 
  respect to the lexicographic product if and only if 
  $H\circ \overline{K}_n + \overline{K}_m$ is prime.
  
  If $H$ has no trivial join-components, then $K_n \circ H \oplus K_m$ is
  prime with respect to the lexicographic product if and only if $H\circ
  K_n \oplus K_m$ is prime.
\end{lem}

Let $H=P\circ \overline{K}_q$ be a PFD of $H$ with $P=\overline{K}_q\circ
A+\overline{K}_m$ such that $A$ has no non-trivial components.  Then $H=
\overline{K}_q\circ R$ with $R=A\circ \overline{K}_q+\overline{K}_m$ is
also a PFD of $H$ that arises by \emph{transposition} of $\overline{K}_q$
from $H=P\circ \overline{K}_q$.  The transposition of $K_q$ is defined
analoguously.
 
\begin{thm}[\cite{GasztImrich71:lexico}]
  \label{thm:UPFDLex}
  Any prime factor decomposition of a graph can be transformed into any
  other one by transpositions of totally disconnected or complete factors.
\end{thm}

\begin{cor}[\cite{GasztImrich71:lexico}]
  \label{cor:UPFDLex}
  All prime factor decompositions of a hypergraph $H$ with respect to the
  lexicographic product have the same number of factors.
  
  If there is a prime factorization of $H$ without complete or totally
  disconnected graphs as factors, then $H$ has a unique prime factor
  decomposition.
  
  If there is a prime factorization of $H$ in which only complete graphs as
  factors have trivial join-components and only totally disconnected
  factors have trivial components, then $H$ has a unique prime factor
  decomposition.
\end{cor}

\subsection{Invariants}

\begin{defn}[Wreath Product of Automorphism Groups
  \cite{DoerflerImrich70:Xsum, GasztImrich71:lexico}]
  Let $H=(V,E)$ and $H'(V',E')$ be hypergraphs. The \emph{wreath product}
  of their automorphism groups is defined as
  \begin{align*}
    Aut(H)\circ Aut(H'):=&\{\varphi\in Aut(H\circ H')\mid
    \exists\,\psi\in Aut(H),\; \forall\, v\in V\,\exists\,\psi'_v\in Aut(H'),\\
    &\text{ s.t. }\varphi((v,v'))=(\psi(v),\psi'_v(v'))\,\forall\,(v,v')\in
    V\times V'\}
  \end{align*}
\end{defn}

Hence $Aut(H)\circ Aut (H')$ forms a subgroup of $Aut(H\circ H')$.
The elements of $Aut(H)\circ Aut (H')$ map $H'$-layer onto $H'$-layer
and are therefore often called \emph{natural} automorphisms.
In \cite{DoerflerImrich70:Xsum} and \cite{Hahn81:lexico} it is shown,
under which conditions holds $Aut(H)\circ Aut(H')=Aut(H\circ H')$.

\begin{thm}[Double Covers \cite{Doerfler78:DoubleCovers}]
  \label{thm:DoubleCovLex}
  If the hypergraphs $H_1, H_2$ are double cover hypergraphs then so is
  their lexicographic product $H_1\circ H_2$.
\end{thm}

\begin{thm}[Hamiltonicity I \cite{Sonntag91:HamTraceLexico}]
  \label{thm:HamILex}
  Let $H_1=(V_1,E_1)$ and $H_2=(V_2,E_2)$, $|V_2|\geq 2$ be two
  hypergraphs.  Then their lexicographic product $H_1\circ H_2$ contains a
  Hamiltonian path if and only if there exists a walk
  $W=(v_0,e_1,v_1,\ldots,e_{k-1},v_k)$ in $H_1$ with $V(W)=V_1$ such that
  $\wp(H_2)\leq \min_{v\in V_1}\left\{|j|\mid v_j=v,
    j\in\{0,\ldots,k\}\right\}$ and $\max_{v\in V_1}\left\{|j|\mid v_j=v,
    j\in\{0,\ldots,k\}\right\}\leq |V_2|$.
\end{thm}

\begin{thm}[Hamiltonicity II \cite{Sonntag91:HamTraceLexico}]
  \label{thm:HamIILex}
  Let $H_1=(V_1,E_1)$ and $H_2=(V_2,E_2)$ be two non-trivial hypergraphs.
  Then their lexicographic product $H_1\circ H_2$ contains a Hamiltonian
  cycle if and only if there exists a walk
  $W=(v_0,e_1,v_1,\ldots,e_{k-1},v_k)$ in $H_1$ with $V(W)=V_1$ such that
  $\wp(H_2)\leq \min_{v\in V_1}\left\{|j|\mid v_j=v,
    j\in\{0,\ldots,k\}\right\}$, $\max_{v\in V_1}\left\{|j|\mid v_j=v,
    j\in\{0,\ldots,k\}\right\}\leq |V_2|$ and $\max_{v\in
    V_1}\left\{|j|\mid v_j=v, j\in\{0,\ldots,k\}\right\}=|V_2|$ implies
  that $v_0\neq v_k$ and there is an edge $e$ in $E_1$ containing both
  $v_0$ and $v_k$.
\end{thm}

\section{Costrong Product}
\subsection{Definition and Basic Properties}

Since the lexicographic product is not commutative, it appears natural to 
consider ``symmetrized'' variants of the lexicographic product. The 
\emph{costrong product}  $H_1 * H_2$ \cite{GasztImrich71:lexico} 
has the edge set 
$$E(H_1 * H_2)= E(H_1\circ H_2)\cup E(H_2\circ H_1).$$
The costrong product is associative, commutative and has $K_1$ as unit
\cite{GasztImrich71:lexico}.  The costrong product of two simple
hypergraphs is not simple, unless both factors are $r$-uniform.  The
projections into the factors are neither (weak) homomorphisms nor preserve
adjacency.  Rank and anti-rank of the costrong hypergraph product $H_1*
H_2$ satisfy
\begin{align*}
  r(H_1*H_2) & =\max\{r(H_1),r(H_2)\}\\
  s(H_1*H_2) & =\min\{s(H_1),s(H_2)\}.
\end{align*}

A hypergraph $H=(V,E)$ is said to be \emph{coconnected} if for each pair of
vertices $u,v\in V$ there exists a sequence of pairwise distinct vertices
$u=u_1,\ldots,u_k=v$ such that consecutive vertices $u_i,u_{i+1}$ are not
both contained in any edge of $H$.  A costrong product of two hypergraphs
is coconnected if and only if both of the factors are.

\begin{thm}[UPFD  \cite{GasztImrich71:lexico}]
  \label{thm:UPFDCo}
  Every finite coconnected hypergraph has a unique PFD w.r.t.\ 
  the costrong product.
\end{thm}

In \cite{Sonntag93:Disjunction} it is shown under which (quite complex)
conditions the costrong product is Hamiltonian. Automorphism groups of
costrong products have been considered by Gaszt and Imrich,
\cite{GasztImrich71:lexico}.

\subsection{Relationships with Graph Products}

The restriction of the costrong product to graphs coincides with the
respective costrong graph product.  The costrong product of graphs can be
obtained from the strong product by virtue of the identity $G_1*G_2 =
\overline{\overline{G_1}\boxtimes \overline{G_2}}$.  This construction is
not applicable to hypergraphs because no suitable definition of complements
of hypergraphs has been proposed so far
\cite{GasztImrich71:lexico-englAbstract, GasztImrich71:lexico}.

The $L2$-section of costrong products can be derived in a straightforward
way from the definition of the $L2$-section of the lexicographic
product. We omit an explicitly description here.

\part{Other Hypergraph Products}

The hypergraph products discussed so far reduce to graph products at least
in the class of simple graphs. In this part of the survey we summarize
alternative constructions that have received considerable attention but do
not correspond to graph products.

\section{The Square Product}
\label{sect:square}

The literature is by no means consistent in its use of the terms ``square
product'' and ``direct product''. In particular, many authors use the term
for ``direct product'' for the square product, see e.g.\
\cite{Berge:DirectProd, Mubayi07:DirectHpg, Pemantle92:DirectHpg,
  DoerrGnewuch06:DiscrSym, DoerrGnewuchHebb05:Discrepancy,
  AhlswedeCai94:Partitioning, Sterboul74:ChromDirect,
  AhlswedeCai97:ExtremalSetPart, Bretto06:Helly,
  Doerfler79:CoversDirectProd, Fuerdi88:Matchings,
  NesestrilRoedl83:SquareProduct, Doerfler82:DirectProd}.  We favor the
term ``square product'' introduced by Ne\v{s}et\v{r}il and R\"{o}dl
\cite{NesestrilRoedl83:SquareProduct}, since it does not reduce to the
direct product on graphs. To add to the confusion, some authors also used
the term ``strong product'', see e.g.\ \cite{Bretto06:Helly}. The
  square product seems to be the most widely studied of the hypergraph
  products.

\subsection{Definition and Basic Properties}

\begin{defn}[Square Product of Hypergraphs \cite{Berge:DirectProd}]
  Given two hypergraphs $H_1=(V_1,E_1)$, $H_2=(V_2,E_2)$
  the \emph{square product} $H=H_1\categ H_2$ has vertex set
  $V(H)=V_1\times V_2$
  and edge set
  \[ E(H)= \left\{e_1\times e_2\mid e_i\in E_i\right\}.\]
\end{defn}

The square product is associative, commutative, distributive w.r.t.\ the
disjoint union and has the single vertex graph with loop $\mathscr{L}K_1$
as unit element.  The square product of two hypergraphs is connected if and
only if both of its factors are connected, and it is a uniform hypergraph
if an only if both of the factors are uniform hypergraphs.  The square
product of two simple hypergraphs is a simple hypergraph.  The projections
into the factors are homomorphisms \cite{Doerfler82:DirectProd}.  For the
rank and anti-rank, respectively, of the square product holds:
\begin{align*}
  r(H_1\categ H_2)=r(H_1)r(H_2)\\
  s(H_1\categ H_2)=s(H_1)s(H_2)
\end{align*}

\begin{prop}[\cite{Berge:DirectProd}]
  The dual hypergraph of a square product of two hypergraphs $H_1$ and
  $H_1$ is the square product of the dual hypergraphs of $H_1$ and $H_2$.
  More formal:
\[(H_1\categ H_2)^*=H_1^*\categ H_2^*\]
\end{prop}

\subsection{Relationships with Graph Products}

The square product of two graphs is not a graph, but a $4$-uniform
hypergraph. Nevertheless, its $2$-section has the structure of a graph
  product.

\begin{lem}
  The $2$-section of the square product 
  $H= H' \categ H''$ is the strong product 
  of the $2$-section of $H'$ and the $2$-section of $H''$, more formally:
  $$[H' \categ H'']_2 =[H']_2 \boxtimes [H'']_2.$$
  \label{lem:2sectionSquare}
\end{lem}
\begin{proof}
  Let $p_1$ and $p_2$ denote $p_{H'}$ and $p_{H''}$, respectively.
  By definition of the $2$-section and the strong graph product, $[H]_2$ and
  $[H']_2 \boxtimes [H'']_2$ have the same vertex set.
  Thus, we need to show that the identity mapping 
  $V([H]_2)\rightarrow V([H']_2 \times [H'']_2)$ is an isomorphism.
  We have:
  $\{x,y\}\in E([H]_2)\Leftrightarrow 
  \exists\, e\in E(H): \{x,y\}\subseteq e = p_1(e)\times p_2(e),\ x\neq y$
  $\Leftrightarrow $
  \begin{enumerate}
  \item $(p_1(x),p_1(y)) \subseteq p_1(e)\in E(H')$ and 
  \item $(p_2(x),p_2(y)) \subseteq p_2(e)\in E(H'')$.
  \end{enumerate}
  Note, in contrast to the other proofs we do not need the condition that
  $p_i(x) \neq p_i(y)$, $i=1,2$. However, since $x\neq y$ we can conclude
  that $p_i(x) \neq p_i(y)$ implies $p_j(x) = p_j(y)$, $i\neq j$.
  Thus, $\{x,y\}\in E([H]_2)$
  $\Leftrightarrow $
  \begin{enumerate}
  \item $(p_1(x),p_1(y)) \in E([H']_2)$ and $p_2(x) = p_2(y)$ or 
  \item $(p_2(x),p_2(y)) \in E([H'']_2)$ and $p_1(x) = p_1(y)$ or 
  \item $(p_1(x),p_1(y)) \in E([H']_2)$ and $(p_2(x),p_2(y)) \in E([H'']_2)$.
  \end{enumerate}
  $\Leftrightarrow $
  $\{x,y\}\in E([H']_2\boxtimes [H'']_2)$.
\end{proof}

\begin{defn}[The Square Product of L2-sections]
  Let $\Gamma_i=(V_i,E'_i,\mathcal{L}_i)$ be the $L2$-section of the
  hypergraphs $H_i=(V_i,E_i)$, $i=1,2.$ The square product of the
  $L2$-sections $\Gamma_1\categ \Gamma_2 = (V,E',\mathcal{L})$ consists
  of the graph $(V,E')=(V_1,E'_1)\boxtimes(V_2,E'_2)$ and a labeling
  function
  $$\mathcal{L}=\mathcal{L}_1\categ \mathcal{L}_2: 
  E'\rightarrow \mathbb{P}(E(H_1\categ H_2))$$
  with 
  $$\mathcal{L}\left(\{(x_1,y_1),(x_2,y_2)\}\right) = 
  \left\{ e_1\times e_2 \mid (A1) \text{ or } (A2) \text{ or } (A3) \right\}$$ 
  where
  \begin{align*}
    & (A1) \quad  x_1=x_2\in e_1 \in E_1, e_2 \in \mathcal{L}_2(\{y_1,y_2\}) \\
    & (A2) \quad  y_1=y_2\in e_2 \in E_2, e_1 \in \mathcal{L}_1(\{x_1,x_2\}) \\
    & (A3) \quad e_1 \in \mathcal{L}_1(\{x_1,x_2\}) \text{ and } 
    e_2 \in \mathcal{L}_2(\{y_1,y_2\})
  \end{align*}	
\end{defn}

One can show that $[H\categ H']_{L_2}=[H]_{L2}\categ [H']_{L2}$ holds for
all simple hypergraphs $H,H'$.

\begin{lem}[Distance Formula]\label{lem:distSquare}
  For all hypergraphs $H,H'$ without loops we have:
  $$d_{H\categ H'}((x,a),(y,b)) =   
  \max\{d_{H}(x,y), d_{H'}(a,b)\}$$
\end{lem}
\begin{proof}
  Combining the results of Lemma \ref{lem:distH2}, Lemma
  \ref{lem:2sectionSquare} and the well-known Distance
  Formula for the strong graph product (Corollary 5.5 in \cite{Hammack:2011a})
  yields to the result.
\end{proof}

\subsection{Prime Factor Decomposition}

\begin{thm}[UPFD  \cite{Doerfler82:DirectProd}]
  \label{thm:UPFDSquare}
  Every finite connected hypergraph (without multiple edges) has
  a unique PFD w.r.t.\ the square product.  
\end{thm}

\subsection{Invariants}

\begin{thm}[Automorphism Group
  \cite{Doerfler82:DirectProd}] \label{thm:AutoDirect} Let
  $H=H_1\categ\ldots\categ H_n$ be the square product of prime
  hypergraphs $H_i$ fulfilling the following condition:
  \begin{itemize}	
  \item[] For each pair of distinct vertices there exists an edge 
    containing exactly one of those vertices.
  \end{itemize}
  Furthermore, let $\varphi\in Aut(H)$.
  Then there exists a permutation $\pi$ of $\{1,\ldots,n\}$ and isomorphisms
  $a_i: H_i\rightarrow H_{\pi(i)}$ such that the $\pi(i)$-th component of 
  $\varphi\left((x_1,\ldots,x_i,\ldots,x_n)\right)$ is $a_i(x_i)$.
\end{thm}

\begin{cor}
  \label{cor:AutoDirect}
  Under the assumptions of Theorem~\ref{thm:AutoDirect} the automorphism
  group of $H=H_1\categ\ldots\categ H_n$ is generated by direct
  products of automorphisms of the $H_i$ and exchanges of isomorphic
  factors.
\end{cor}

\begin{thm}[$k$-fold Covering \cite{Doerfler79:CoversDirectProd}]
  \label{thm:CovSquare}
  Let $H_i'=(V_i',E_i')$ be a $k_i$-fold covering of hypergraphs
  $H_i=(V_i,E_i)$ via a covering projection $p'_i$, $i=1,2$.  Then the
  square product $H'_1\categ H'_2$ is a $k_1k_2$-fold cover of $H_1
  \categ H_2$ via a covering projection $\pi$ induced naturally by
  $p'_1$ and $p'_2$, i.e, define $p$ by:
\begin{itemize}
\item[] $p((x,y)) = (p'_1(x),p'_2(y))$, for $(x,y) \in V'_1\times V'_2$, 
\item[] $p(e_1\times e_2) = p'_1(e_1)\times p'_2(e_2)$, 
  for $e_1 \in E'_1,\, e_2\in E'_2$, 
\end{itemize}
\end{thm}

\begin{thm}[Conformal Hypergraphs \cite{Bretto06:Helly}]
  \label{thm:ConfSquare}
  $H = H_1 \categ H_2$ is conformal if and only if $H_1$ and $H_2$ are
  conformal.
\end{thm}

\begin{thm}[Helly Property \cite{Bretto06:Helly}]
  \label{thm:HellySquare}
  $H = H_1 \categ H_2$ has the Helly property if and only if $H_1$ and
  $H_2$ have the Helly property.
\end{thm}

\begin{thm}[Stability Number \cite{Sterboul74:ChromDirect, Berge:DirectProd}]
  \label{thm:StabSquare}
  For any two hypergraphs $H=(V,E)$ and $H'=(V',E')$ with stability number 	
  $\beta$ and $\beta'$, respectively,
  holds
  $$\beta\beta'+|V'|\beta + |V|\beta' \leq\beta(H\categ H') 
  \leq |V'|\beta+ |E|\beta'.$$
\end{thm}

The stability number $\beta(H)$ of a hypergraph $H$ satisfies $\beta(H) =
|V(H)| - \tau(H)$ and $\beta(H\categ H') = |V(H)||V(H')| -
\tau(H\categ H')$. Further results for the stability number can thus be
obtained from the properties of the covering number
\cite{Berge:DirectProd}.

\begin{thm}[Matching and Covering \cite{Berge:Hypergraphs}]	
  \label{thm:MatchCov}
  For two hypergraphs $H$ and $H'$ we have
  \begin{align*}
    \nu(H)\nu(H')\leq\nu(H\categ H')\leq\tau^*(H)\nu(H')\leq
    \tau^*(H)\tau^*(H')\\
    =\tau^*(H\categ H')\leq\tau^*(H)\tau(H')\leq
    \tau(H\categ H')\leq\tau(H)\tau(H').
  \end{align*}
\end{thm}

\begin{thm}[Fractional Covering Number I \cite{Berge:DirectProd}]	
  \label{thm:SqProdCovI}
  A necessary and sufficient condition for a hypergraph $H$ to satisfy
  $\tau(H\categ H') = \tau(H)\tau(H')$ for all $H'$ is that 
  $\tau(H) = \tau^*(H)$.
\end{thm}

\begin{thm}[Fractional Covering Number II \cite{Berge:DirectProd}]	
  \label{thm:SqProdCovII}
  Let $H$ and $H'$ be two hypergraphs. Then
  $$\tau(H\categ H') \geq \tau(H) +\tau(H') - 1.$$
  A hypergraph $H=(V,E)$ satisfies $\tau(H\categ H') = 
  \tau(H) +\tau(H') - 1$
  for every hypergraph $H'$ if and only if
  $\bigcap_{e\in E}\neq\emptyset$.
\end{thm}

\begin{thm}[Fractional Covering Number III \cite{McEP-1971}] 
  \label{thm:SqProdCovIII} 
  For every hypergraph $H$ holds:
  $$ \tau^*(H)  = \lim_{n \to \infty} \sqrt[n]{\tau(H^{\categ n})}$$
\end{thm}

\begin{thm}[Fractional Covering Number IV 
  \cite{Fuerdi88:Matchings, Berge:Hypergraphs}]	\label{thm:SqProdCovVI}
  For every hypergraph $H$ holds:
  $$\tau^*(H) = \min_{H'} \frac{\tau(H\categ H')}{\tau(H')}, $$
  where $H'$ runs over all hypergraphs.
\end{thm}

\begin{thm}[Fractional Covering and Matching Number  
  \cite{Berge:Hypergraphs}] \label{thm:SqProdCovMatch}
  For every hypergraph $H$ with Helly property holds:
  $$\tau^*(H) = \min_{H'} \frac{\tau(H\categ H')}{\tau(H')}, $$
  where $H'$ runs over all hypergraphs.
\end{thm}

\begin{thm}[Partition Number \cite{AhlswedeCai94:Partitioning}] 
  \label{thm:SqProdPart}
  Let $d\in \mathbb{N}$ and 
  $H_i = (V_i,E_i), i=1,\dots,n$ be hypergraphs such that
  $E_i = \binom{V_i}{d} \cup\{\{v\}\mid v\in V_i\}$. Moreover, 
  let $r_i  = |V_i|\mod d$. If $d> \prod_{i:r_i\neq 0} r_i $
  then $$\rho(\categ_{i=1}^n H_i)  =  \prod_{i=1}^n \rho(H_i).$$
\end{thm}

\begin{thm}[Chromatic Number I \cite{Berge:DirectProd}]	
  \label{thm:SqProdColI}
  For two hypergraphs $H$ and $H'$ we have:
  $$\chi(H\categ H') \leq \min\{\chi(H),\chi(H')\}.$$
\end{thm}

In \cite{Berge:DirectProd} the authors asked whether the chromatic number
of the square product of two hypergraphs goes to infinity if the chromatic
numbers of both of the factors go to infinity.
The following theorem, which is due to D. Mubayi and V. R\"{o}dl, refutes this conjecture. 

\begin{thm}[Chromatic Number II \cite{Mubayi07:DirectHpg}] 
  \label{thm:SqProdColII}
  For every integer $n\geq 2$ and $k=2^n$ there exists a hypergraph $H_k$
  satisfying $\chi(H_k) > k$ and $\chi(H_k\categ H_k) = 2$.
\end{thm}

Moreover, in \cite{Mubayi07:DirectHpg} the authors conjectured that for every $r\geq 2$
there is a $c\geq 2$ such that for every positive integer $k$ there exists $r$-uniform
hypergraphs $G_k$ and $H_k$ for which $\chi(G_k) > k$, $\chi(H_k) > k$ and $\chi(G_k\categ H_k) \leq c$,
and that this also true for the special case $c=r=2$.

Other results concerning the chromatic number of a special case of square products, 
i.e., square products of complete graphs, are due to
Sterboul and can be found in \cite{Sterboul74:ChromDirect}. 

\begin{thm}[Discrepancy  \cite{doerr:discrepancy}] \label{thm:SqProdColIII}
  For any $k\in \mathbb{N}$ and any two hypergraphs $H,H'$ it holds:
  $$ \disc(H \categ H',k) \leq k\cdot \disc(H,k)\cdot \disc(H',k).$$
\end{thm}

A special partial hypergraph of the square product $H^{\categ d}$ 
of a hypergraph $H =(V,E)$ has found particular attention
and is also called the \emph{$d$-fold symmetric product}, defined as
the subgraph $\{V^d, \{e^d, e\in E\}\}$, where $M^d$ denotes the
usual $d$-fold Cartesian set product of the set $M$.
In \cite{DoerrGnewuch06:DiscrSym, DoerrGnewuchHebb05:Discrepancy} 
the authors gave several upper and lower bounds for the discrepancy 
w.r.t.\ this product.

\section{The Categorial Product}

The following hypergraph product was motivated by the investigation of a
category of hypergraphs \cite{Doerfler80:CategoriesHpg}.  It is categorical
in the category of hypergraphs. It has rarely been studied since its
introduction, however; to our knowledge, the only systematic account is a
contribution by X.\ Zhu \cite{Zhu92:DirectProd}.

\subsection{Definition and Basic Properties}

\begin{defn}[Hypergraph Product \cite{Doerfler80:CategoriesHpg}]
  \label{defi:CatProd}
  Let $H_i=(V_i,E_i)$, $i=1,2$ be two hypergraphs.
  Then their \emph{product}
  $H=H_1\timesH H_2$ has edge set $V(H)=V_1\times V_2$ and vertex set
  \begin{align*}
    E(H)=\left\{e\subseteq V(H)\mid p_i(e)\in E_i,\; i=1,2\right\}
  \end{align*}
\end{defn}

The categorial hypergraph product is associative, commutative, distributive
w.r.t.\ the disjoint union and has the single vertex graph with loop
$\mathscr{L}K_1$ as unit element.  The projections into the factors are, by
definition, homomorphisms.  However, the product of two simple hypergraphs
is not a simple hypergraph and the product of two non-trivial uniform
hypergraphs does not result in a uniform hypergraph.  For the rank and
anti-rank, respectively, of a hypergraph product holds:
\begin{align*}
  r(H_1\timesH H_2) &=r(H_1)r(H_2)\\
  s(H_1\timesH H_2) &=\max\{s(H_1),s(H_2)\}
\end{align*}

We have the following relations with other hypergraph products:
\begin{itemize}
\item $E(H_1\weitdacherl{\times} H_2) = E'\subseteq E(H_1\timesH H_2)$ with \\
  $E':= \{e\in  E(H_1\timesH H_2) \mid   \nexists e'\in E(H_1\timesH H_2) 
  \text{ s.t. } e'  \subset e \}$
  for simple hypergraphs $H_1$ and $H_2$ \cite{Zhu92:DirectProd}.
\item $E(H_1\categ H_2)\subseteq E(H_1\timesH H_2)$
\end{itemize}

\subsection{Relationships with Graph Products}
The restriction of this product to graphs does not coincide with any known
graph product.  Moreover, the product $G_1\timesH G_2$ of two graphs $G_1$
and $G_2$ is no graph anymore, but a hypergraph of rank $4$.

\begin{lem}
  \label{lem:2sectionCateg}
  The $2$-section of the product 
  $H= H' \timesH H''$ is the strong product 
  of the $2$-section of $H'$ and the $2$-section of $H''$, more formally:
  $$[H' \timesH H'']_2 =[H']_2 \boxtimes [H'']_2.$$
\end{lem}
This result can be proved  analogously to the proof of 
Lemma \ref{lem:2sectionSquare}.
  
\begin{defn}[The Categorial Product of L2-sections]
  Let $\Gamma_i=(V_i,E'_i,\mathcal{L}_i)$ be the $L2$-section of the
  hypergraphs $H_i=(V_i,E_i)$, $i=1,2$.  The categorial product of the
  $L2$-sections $\Gamma_1\timesH \Gamma_2 = (V,E',\mathcal{L})$ consists of
  the graph $(V,E')=(V_1,E'_1)\boxtimes(V_2,E'_2)$ and a labeling function
  $$\mathcal{L}=\mathcal{L}_1\timesH \mathcal{L}_2: E'\rightarrow 
  \mathbb{P}(E(H_1\timesH H_2))$$
  assigning to each edge $e' = \{(x_1,y_1),(x_2,y_2)\} \in E'$ a label
  with 
  $$\mathcal{L}\left(e'\right) =
  \begin{cases}
    \bigcup_{e_1\in \mathcal{L}(e'_1), e_2\in E_2:y_1\in e_2} 
    \mathcal{E}(e',e_1,e_2), 
    &\text{ if } y_1=y_2 \\
    \bigcup_{e_2\in \mathcal{L}(e'_2), e_1\in E_1:x_1\in e_1} 
    \mathcal{E}(e',e_1,e_2), 
    &\text{ if } x_1 =  x_2\\
    \bigcup_{e_1\in \mathcal{L}(e'_1), e_2\in \mathcal{L}(e'_2)} 
    \mathcal{E}(e',e_1,e_2), 
    &\text{ else, i.e., } x_1\neq x_2, y_1\neq y_2 
  \end{cases}	
  $$
  where $$\mathcal{E}(e',e_1,e_2) = \left\{ f\subseteq e_1\times e_2 \mid 
    e'\subseteq f, p_i(f)=e_i, i=1,2\right\}.$$
\end{defn}

The identity $[H\timesH H']_{L_2}=[H]_{L2}\timesH [H']_{L2}$ holds for all
simple hypergraphs $H,H'$.

\begin{lem}[Distance Formula]\label{lem:distCateg}
  For all hypergraphs $H,H'$ without loops we have:
  $$d_{H\timesH H'}((x,a),(y,b)) =   
  \max\{d_{H}(x,y), d_{H'}(a,b)\}$$
\end{lem}
\begin{proof}
  Combining the results of Lemma \ref{lem:distH2}, Lemma
  \ref{lem:2sectionCateg} and the well-known Distance
  Formula for the strong graph product (Corollary 5.5 in \cite{Hammack:2011a})
  yields to the result.
\end{proof}

\subsection{Invariants}

Only the chromatic number of the categorial product $H_1\timesH H_2$ has
been investigated in some detail \cite{Zhu92:DirectProd}.

\begin{thm}[Chromatic Number I \cite{Zhu92:DirectProd}]
\label{thm:ColCatI}
Let $H_1$ and $H_2$ be two hypergraphs such that $\chi(H_i) = n+1$.
Moreover, let $H_i$, $i=1,2$ contain a partial hypergraph $H_i' =
\{V_i',\{e\subseteq V_i'\mid |e|\geq 2\}\}$ with $|V_i'|=n$ and a
vertex-critical $n+1$ chromatic partial hypergraph $H_i'' = (V_i'',E_i'')$,
i.e., for any vertex $x\in V_i''$ the hypergraph induced by $V_i''-\{x\}$
is $n$-colorable. Furthermore, let $V_i'\cap V_i'' \neq\emptyset$, $i=1,2$.
Then $\chi(H_1 \timesH H_2)=n+1$.
\end{thm}

\begin{thm}[Chromatic Number II \cite{Zhu92:DirectProd}]
\label{thm:ColCatII}
Let $H=(V,E)$ be a hypergraph with $\chi(H)= n+1$ such that any $v\in V$ is
contained in a partial hypergraph $H' = (V',E')$ of $H$ with $|V'|=n$ and
$E' = \{e\subseteq V'\mid |e|\geq 2\}$.  Then for any hypergraph $G$ with
$\chi(G)=n+1$ holds $\chi(H \timesH G)=n+1$.
\end{thm}

\part{Beyond Finite and Undirected Hypergraphs}
\label{part:further aspects}

\section{Infinite Hypergraphs}

Only finite hypergraphs and products of finitely many factors have been
treated so far.  It is possible to extend the definitions of the products
to infinitely many finite and to infinite hypergraphs. For this purpose we
need the following definition.  For an arbitrary family of (vertex) sets
$V_i$, $i\in I$, their Cartesian set product $V=\bigtimes_{i\in I}V_i$
consists of the set of all functions $x: i\mapsto x_i$, $x_i\in V_i$ of $I$
into $\bigcup_{i\in I} V_i$.  Notice, that the Cartesian set product of an
arbitrary family of sets is usually denoted by $\prod_{i\in I}V_i$, but to
emphasize the relation to the finite case, we use the term $\bigtimes_{i\in
  I}V_i$ instead.  In this case, the \emph{projection} $p_j:V\to V_j$ is
defined by $v \mapsto v_j$ whenever $v: j\mapsto v_j$.  As before, we will
call $v_j$ the \emph{$j$-th coordinate} of the vertex $v\in V$.

Several of the hypergraph products discussed in the previous section are
connected provided each of the the finitely many finite factors are
connected. A corresponding result can be established for finitely many
connected factors of infinite size using the Distance Formula for the
respective product. In contrast, connectedness results do not necessarily
carry over to products of infinitely many hypergraphs. As an example
consider the Cartesian product of infinitely many factors.  There are
vertices that differ in infinitely many coordinates and thus, by the
Distance Formula cannot be connected by a path of finite length,
\cite{OstHellmStad11:CartProd, Imrich70:SchwachKartProd}. This in turn leads
to difficulties concerning the prime factor decomposition.  Again, consider
the Cartesian product.  An infinite connected hypergraph can have
infinitely many prime factors.  In this case it cannot be the Cartesian
product of these factors, since the product is not connected, but a
connected component of this product.  For this purpose, the \emph{weak
  Cartesian product} is presented which was first introduced by Sabidussi
\cite{Sabidussi60:GraphMult} for graphs and later generalized by Imrich
\cite{Imrich70:SchwachKartProd}.

Let $\{H_i\mid i\in I\}$ be a family of hypergraphs and let $a_i\in V(H_i)$
for $i\in I$.  The weak Cartesian product $H=\Box_{i\in I}(H_i,a_i)$ of the
rooted hypergraphs $(H_i,a_i)$ is defined by
\begin{align*}
  V(H)&=\{v\in\bigtimes_{i\in I}V(H_i)\mid p_i(v)\neq a_i
  \text{ for at most finitely many }i\in I\}\\
  E(H)&=\{e\subseteq V(H)\mid p_j(e)\in E(H_j)
  \text{ for exactly one }j\in I, 
  |p_i(e)|=1\text{ for }i\neq j\}.
\end{align*}

Note the weak Cartesian product does not depend on the ``reference
coordinates'' $a_i$; furthermore, it reduced to the ordinary Cartesian
product if $I$ is finite.  The weak Cartesian product is associative,
commutative, distributive w.r.t.\ the disjoint union and has trivially $K_1$
as unit. Furthermore, the weak Cartesian product of connected hypergraphs
is connected \cite{Imrich70:SchwachKartProd}.

\begin{thm}[UPFD \cite{OstHellmStad11:CartProd}]
  Every simple connected finite or infinite hypergraph 
  with finitely or infinitely many factors 
  has a unique PFD w.r.t.\ the weak Cartesian product.
\end{thm}
We suspect that similar constructions can be used to define infinite
versions of the other hypergraph products treated in this contribution.

Other results for infinite hypergraphs are known e.g. about the chromatic
number of square products:
\begin{thm}[\cite{Mubayi07:DirectHpg}]
  \label{thm:chromInfty}
  Let $H$ and $H'$ be two hypergraph whose edges have finite size. Suppose
  that $\chi(H) = \chi(H') = \infty$. Then $\chi(H\categ H') = \infty$.
\end{thm}   

Surprisingly, Theorem \ref{thm:chromInfty} does not hold if both hypergraphs
have edges that are all of infinite size. Let $V(H) = \{1,2,\dots \}$ and
$E(H)$ comprising all infinite subsets of $V(H)$. Clearly, 
$\chi(H) = \infty$, since any coloring with finitely many colors results
in an edge colored with only one color. 
As mentioned in \cite{Mubayi07:DirectHpg}, in $H\categ H$ there exists
a proper $2$-coloring assigning each vertex $(i,j)\in V(H\categ H)$
one color if $i\leq j$ and the other color else.  

\begin{thm}[\cite{Mubayi07:DirectHpg}]
  Let $H$ and $H'$ be two hypergraph whose edges are all of infinite size. 
  Suppose that $\chi(H) = \chi(H') = \infty$. 
  Then $\chi(H\categ H') = 2$.
\end{thm}   

\section{Directed Hypergraphs}

Directed hypergraphs play a role e.g.\ as models of chemical reaction
networks and in transit and satisfiability problems, see
\cite{Gallo:98,Zeigarnik:00a,Ausiello:01} for reviews. Directed hypergraphs
can be defined in various ways.  Here, we refer to the most general
definition.  A \emph{directed hypergraph} $\oH=(V,\oE)$ consists of a
vertex set $V$ and a set of \emph{hyperarcs} $\oE$, where each hyperarc
$\oe\in\oE$ is an ordered pair of nonempty, not necessarily disjoint
subsets of $V$, $\oe=(t(e),h(e))$, the \emph{tail} and \emph{head} of
$\oe$, respectively.  We call a directed hypergraph \emph{simple}, if
$t(e)\subseteq t(e')$ and $h(e)\subseteq h(e')$ implies $\oe=\oe'$.  The
$2$-section is then a directed graph with arc $(x,y)$ if there is an edge
$(t(e),h(e))$ with $x\in t(e)$ and $y\in h(e)$.

Product structures have not been studied extensively in a directed setting,
even though there are some exceptions.  The lexicographic product of
directed graphs, for instance, appears in a general technique to amplify
lower bounds for index coding problems \cite{BKL11:NetworkCoding}.

The directed version of the Cartesian product was first introduced in
\cite{OstHellmStad11:CartProd}.  The Cartesian product $\oH=\oH_1\Box\oH_2$
of two directed hypergraphs $\oH_1=(V_1,\oE_1), \oH_2=(V_2,\oE_2)$ has
edge set
\begin{align*}
  \oE(\oH)=&\left\{\left(\{x\}\times t(f),\{x\}\times h(f)\right)\mid
    x\in V_1,\of\in\oE_2\right\}\\
  &\cup\left\{\left(t(e)\times\{y\},h(e)\times\{y\}\right)\mid
    \oe\in\oE_1,y\in V_2\right\}.
\end{align*} 
Basic properties of the Cartesian product of undirected hypergraphs can
immediately be transferred to the directed case. Moreover, uniqueness of
the PFD was shown in \cite{OstHellmStad11:CartProd}.

\begin{thm}[UPFD \cite{OstHellmStad11:CartProd}]
  Every connected (finite or infinite) directed hypergraph has a unique PFD
  w.r.t. the (weak) Cartesian product.
\end{thm}

The definition of the square product might be transferred to hypergraphs as
follows: The square product $\oH=\oH_1\categ\oH_2$ of two directed
hypergraphs $\oH_1=(V_1,\oE_1), \oH_2=(V_2,\oE_2)$ has edge set
\begin{align*}
  \oE(\oH)=\left\{\left(t(e)\times t(f),h(e)\times h(f)\right)\mid
    \oe\in \oE_1,\of \in\oE_2\right\}
\end{align*} 

In \cite{Imrich:06a} the authors introduce the square product of so-called
$\mathcal{N}$-systems, that is a special class of directed hypergraphs.
More precisely, $\oN=(V,\oE)$ is an $\mathcal{N}$-system if $|t(e)|=1$ and
$t(e)\subseteq h(e)$ holds for all $\oe\in\oE$ and
$\bigcup_{\oe\in\oE}t(e)=V$.  It is shown that the square product is
closed in the class of $\mathcal{N}$-systems, i.e., the square product of
two $\mathcal{N}$-systems is again an $\mathcal{N}$-system.

\begin{thm}[\cite{Imrich:06a}]
  Let $\oN$ be an $\mathcal{N}$-system.  If $[\oN]_2$ is thin, i.e., there
  are no two vertices $u,v\in V(\oN)$ with $(u,x)\in \oE([\oN]_2)
  \Leftrightarrow (v,x)\in\oE([\oN]_2)$ and $(x,u)\in \oE([\oN]_2)
  \Leftrightarrow (x,v)\in\oE([\oN]_2)$, and connected, then $\oN$ has a
  unique PFD with unique coordinatization.
\end{thm}
The authors conjectured, furthermore, that the condition of thinness can be
omitted as long as one is satisfied with a unique PFD without insisting on
a unique coordinatization.

\section{Summary}

Table~\ref{tab:properties} provides an overview of the properties of the
hypergraph products discussed in this survey. 
Table~\ref{tab:invariants} shows which hypergraph invariants can be transferred 
from factors to products at least under some additional conditions. 

We considered the following properties:
\begin{itemize}
\item[(P1)] Associativity.
\item[(P2)] Commutativity.
\item[(P3)] Distributivity with respect to the disjoint union.
\item[(P4)] Existence of a unit.
\item[(P5)] $H_1\circledast H_2$ is (co-)connected 
  $\Leftrightarrow $$H_1$ and $H_2$ are (co-)connected.
\item[(P6)] If $H_1$ and $H_2$ are simple then $H_1\circledast H_2$ is simple.
\item[(P7)] The projections $p_i:V(H_1\circledast H_2)\rightarrow V(H_i)$ 
  for $i\in \{1,2\}$ are (at least weak) homomorphisms.
\item[(P8)] The projections preserve adjacency.
\item[(P9)] The adjacency properties of a product depends on those 
  of its factors.
\item[(P10)] Unique prime factorization in special classes of hypergraphs.
\item[(P11)] Preserves uniformity
\item[(P12)] Preserves $r$-uniformity
\item[(P13)] The restriction of the product $\circledast$ on simple graphs 
             coincides with the respective graph product. 
\item[(P14)] The restriction of the product $\circledast$ on not necessarily 
  simple graphs is the corresponding graph product. 
\item[(P15)] The $2$-section of the product coincides with the graph 
  product of the $2$-section of the factors.
\end{itemize}

\begin{table}[H]
\small
	\centering
	\begin{tabular}{|l|c|c|c|c|c|c|c|c|c|c|}
		\hline
											&\large{$\Box$}	&\large{$\weithutzl{\times}$}	&\large{$\weitdacherl{\times}$}	&\large{$\widetilde{\times}$}
											&\large{$\weithutzl{\boxtimes}$}	&\large{$\weitdacherl{\boxtimes}$}	&\large{$\circ$}	&\large{$*$}
											&\large{$\categ$} &\large{$\circledcirc$}\\
		\hline
		P1																		&$\surd$	&$\surd$	&$\surd$	&$\surd$	&$\surd$	
																					&$\surd$	&$\surd$	&$\surd$	&$\surd$	&$\surd$\\
		\hline
		P2																		&$\surd$	&$\surd$	&$\surd$	&$\surd$	&$\surd$	
																					&$\surd$	&$-$	&$\surd$	&$\surd$	&$\surd$\\
		\hline
		P3																		&$\surd$	&$\surd$	&$\surd$	&$\surd$	&$\surd$	
																					&$\surd$	&$-$	&$-$	&$\surd$	&$\surd$\\
		\hline
		P4																		&$K_1$	&$-$	&$\mathscr{L}K_1$	&$-$	&$K_1$	
																					&$K_1$	&$K_1$	&$K_1$	&$\mathscr{L}K_1$	&$\mathscr{L}K_1$\\
		\hline
		P5																		&$\Leftrightarrow$	&$\Rightarrow$	&$\Rightarrow$	&$\Rightarrow$				
																					&$\Leftrightarrow$	&$\Leftrightarrow$	&iff $H_1$	&$\overset{co}{\Leftrightarrow}$	
																					&$\Leftrightarrow$	&$\Leftrightarrow$\\
		\hline
		P6																		&$\surd$	&$\surd$	&$\surd$	&$\surd$	&$\surd$	
																					&$\surd$	&$\surd$	&$-$	&$\surd$	&$-$\\
		\hline
		P7																		&w	&$-$	&$\surd$	&$\surd$	&$-$	
																					&w	&$p_1$, w	&$-$	&$\surd$	&$\surd$\\
		\hline
		P8																		&$\surd$	&$\surd$	&$\surd$	&$\surd$	&$\surd$	
																					&$\surd$	&$p_1$	&$-$	&$\surd$	&$\surd$\\
		\hline
		P9																		&$\surd$	&$\surd$	&$\surd$	&$\surd$	&$\surd$	
																					&$\surd$	&$\surd$	&$\surd$	&$\surd$	&$\surd$\\
		\hline
		P10																		&\cref{thm:UPFDCart},\ref{thm:UPFDCartAlg}	&?	&?	&?	&?	
																					&?	&\ref{thm:UPFDLex},\ref{cor:UPFDLex}	&\ref{thm:UPFDCo}
																					&\ref{thm:UPFDSquare}	&?\\
		\hline
		P11																		&$\surd$	&$\surd$	&$\surd$	&$\surd$	&$\surd$	
																					&$\surd$	&$\surd$	&$\surd$	&$\surd$	&$-$\\
		\hline
		P12																		&$\surd$	&$\surd$	&$\surd$	&$-$	&$\surd$	
																					&$\surd$	&$\surd$	&$\surd$	&$-$	&$-$\\
		\hline
		P13																		&$\surd$	&$\surd$	&$\surd$	&$\surd$	&$\surd$	
																					&$\surd$	&$\surd$	&$\surd$	&$-$	&$-$\\
		\hline
		P14																		&$\surd$	&$-$	&$\surd$	&$-$	&$-$	
																					&$\surd$	&$\surd$	&$\surd$	&$-$	&$-$\\
		\hline
		P15																		&$\Box$	&$\times$	&$-$	&$-$	&$\boxtimes$	
																					&$\boxtimes$	&$\circ$	&$*$	&$\boxtimes$	&$\boxtimes$\\
                                                                                                                                                                        \hline
\end{tabular}
\normalsize
\caption[Properties of the hypergraph products]{Properties of the 
  hypergraph products. }

\label{tab:properties}
\end{table}

We considered the following invariants:
\begin{itemize}
\item[(I1)] Automorphism group
\item[(I2)] $k$-fold covering
\item[(I3)] Independence, matching and covers
\item[(I4)] Coloring properties
\item[(I5)] Helly property
\item[(I6)] Hamiltonicity
\end{itemize}

\begin{table}[H]
	\centering
	\small

\begin{tabular}{|l|c|c|c|c|c|c|c|c|}
		\hline
											&\large{$\Box$}
&\large{$\weithutzl{\times}$},\large{$\weitdacherl{\times}$},\large{$\widetilde{\times}$},\large{$\weitdacherl{\boxtimes}$}	
											&\large{$\weithutzl{\boxtimes}$}	&\large{$\circ$}	&\large{$*$}
											&\large{$\categ$} &\large{$\circledcirc$}\\
		\hline
		I1																		&\ref{thm:AutCart}	&?	&?	
																					&\cite{DoerflerImrich70:Xsum, Hahn81:lexico}	&\cite{GasztImrich71:lexico}	
																					&\ref{thm:AutoDirect},\,\ref{cor:AutoDirect}	&?\\
		\hline
		I2																		&\ref{thm:k-covCart}	&?	&?	
																					&\ref{thm:DoubleCovLex}	&?	&\ref{thm:CovSquare}	&?\\
		\hline
		I3																		&?	&?	&?	
																					&?	&?	
																					&\cref	{thm:StabSquare,thm:MatchCov,thm:SqProdCovI,thm:SqProdCovII,thm:SqProdCovIII,thm:SqProdCovVI,thm:SqProdCovMatch,thm:SqProdPart}	&?\\
		\hline
		I4																		&\ref{thm:ColorICart},\,\ref{thm:ColorIICart}	&?	&?	
																					&?	&?	&\cref{thm:SqProdColI,thm:SqProdColII,thm:SqProdColIII}	
																					&\cref{thm:ColCatI},\,\ref{thm:ColCatII}\\
		\hline
		I5																		&\ref{thm:ConfCart},\,\ref{thm:HellyCart}	&?	&?	
																					&?	&?	&\ref{thm:ConfSquare},\,\ref{thm:HellySquare}	&?\\
		\hline
		I6																		&\ref{thm:HamICart},\,\ref{thm:HamIICart}	&?	
																					&\ref{thm:StrongProdHamI},\,\ref{thm:StrongProdHamII}	
																					&\ref{thm:HamILex},\,\ref{thm:HamIILex}	&\cite{Sonntag93:Disjunction}	
																					&? &?\\
		\hline
\end{tabular}

\normalsize
\caption[Invariants]{Invariants and Hypergraph Products}	
\label{tab:invariants}
\end{table}
\normalsize

 The two summary tables use symbol ``$\surd$'' to indicate that a
  condition is satisfied, while ``$-$'' means that the product does not
  have the property in question. The question mark ``?'' implies that it is
  unknown at present whether a particular statement is true.  Numbers in
  brackets refer to citations, while numbers without brackets refer to
  theorems listed in this survey that establish the property under certain
  additional preconditions or provides results on particular invariants.
  If (P7) holds only for weak homomorphisms we indicate this with ``w''.
  If (P7) and (P8) holds only for the projection onto the first factor we
  indicate this by ``$p_1$''.

\bibliographystyle{plain}
\bibliography{biblio}

\begin{thebibliography}{10}

\bibitem{AhlswedeCai94:Partitioning}
R.~Ahlswede and N.~Cai.
\newblock {On partitioning and packing products with rectangles.}
\newblock {\em Comb. Probab. Comput.}, 3(4):429--434, 1994.

\bibitem{AhlswedeCai97:ExtremalSetPart}
R.~Ahlswede and N.~Cai.
\newblock {On extremal set partitions in Cartesian product spaces.}
\newblock {Bollob\'as, B\'ela (ed.) et al., Combinatorics, geometry and
  probability. A tribute to Paul Erd\H{o}s. Proceedings of the conference
  dedicated to Paul Erd\H{o}s on the occasion of his 80th birthday, Cambridge,
  UK, 26 March 1993. Cambridge: Cambridge University Press. 23-32 (1997).},
  1997.

\bibitem{AMA-07}
D.~Archambault, T.~Munzner, and D.~Auber.
\newblock {T}opo{L}ayout: Multilevel graph layout by topological features.
\newblock {\em IEEE Transactions on Visualization and Computer Graphics},
  13(2):305--317, 2007.

\bibitem{Ausiello:01}
G.~Ausiello, P.~G. Franciosa, and D.~Frigioni.
\newblock Directed hypergraphs: problems, algorithmic results, and a novel
  decremental approach.
\newblock In A.~Restivo, S.~R.~D. Rocca, and L.~Roversi, editors, {\em ICTCS},
  volume 2202 of {\em Lecture Notes in Computer Science}, page 312–327.
  Springer, 2001.

\bibitem{Bandelt:91Helly}
H.-J. Bandelt and Erich Prisner.
\newblock Clique graphs and helly graphs.
\newblock {\em J. Comb. Theory, Ser. B}, 51(1):34--45, 1991.

\bibitem{Berge:Hypergraphs}
C.~Berge.
\newblock {\em Hypergraphs: {C}ombinatorics of finite sets}, volume~45.
\newblock North-Holland, Amsterdam, 1989.

\bibitem{Berge:DirectProd}
C.~Berge and M.~Simonovitis.
\newblock The coloring numbers of the direct product of two hypergraphs.
\newblock In {\em Hypergraph {S}eminar}, volume 411 of {\em Lecture Notes in
  Mathematics}, pages 21--33, Berlin / Heidelberg, 1974. Springer-Verlag.

\bibitem{BKL11:NetworkCoding}
A.~Blasiak, R.~Kleinberg, and E.~Lubetzky.
\newblock Lexicographic products and the power of non-linear network coding.
\newblock {\em CoRR}, abs/1108.2489, 2011.

\bibitem{Bretto06:Helly}
A.~Bretto.
\newblock Hypergraphs and the helly property.
\newblock {\em Ars Comb.}, 78:23--32, 2006.

\bibitem{BrettoSilvestre10:Factorization}
A.~Bretto and Y.~Silvestre.
\newblock {Factorization of Cartesian products of hypergraphs.}
\newblock {Thai, My T. (ed.) et al., Computing and combinatorics. 16th annual
  international conference, COCOON 2010, Nha Trang, Vietnam, July 19--21, 2010.
  Proceedings. Berlin: Springer. Lecture Notes in Computer Science 6196,
  173-181 (2010).}, 2010.

\bibitem{Bretto09:HyperCartProd}
A.~Bretto, Y.~Silvestre, and T.~Vall{\'e}e.
\newblock Cartesian product of hypergraphs: properties and algorithms.
\newblock In {\em 4th Athens Colloquium on Algorithms and Complexity (ACAC
  2009)}, volume~4 of {\em EPTCS}, pages 22--28, 2009.

\bibitem{Cupal:00a}
J.~Cupal, S.~Kopp, and P.~F. Stadler.
\newblock {RNA} shape space topology.
\newblock {\em Artificial Life}, 6:3--23, 2000.

\bibitem{DoerrGnewuchHebb05:Discrepancy}
B.~Doerr, M.~Gnewuch, and N.~Hebbinghaus.
\newblock {Discrepancy of products of hypergraphs.}
\newblock {Felsner, Stefan (ed.), 2005 European conference on combinatorics,
  graph theory and applications (EuroComb '05). Extended abstracts from the
  conference, Technische Universit\"at Berlin, Berlin, Germany, September 5--9,
  2005. Paris: Maison de l'Informatique et des Math\'ematiques Discr\`etes
  (MIMD). Discrete Mathematics \&amp; Theoretical Computer Science.
  Proceedings. AE, 323-328, electronic only (2005).}, 2005.

\bibitem{DoerrGnewuch06:DiscrSym}
B.~Doerr, M.~Gnewuch, and N.~Hebbinghaus.
\newblock {Discrepancy of symmetric products of hypergraphs.}
\newblock {\em Discr. Math. Theor. Comp. Sci.}, AE:323--238, 2005.

\bibitem{doerr:discrepancy}
B.~Doerr, A.~Srivastav, and P.~Wehr.
\newblock Discrepancy of cartesian products of arithmetic progressions.
\newblock {\em Electr. J. Comb.}, 11s:1--16, 2004.

\bibitem{Doerfler78:DoubleCovers}
W.~D{\"o}rfler.
\newblock {Double covers of hypergraphs and their properties.}
\newblock {\em Ars Comb.}, 6:293--313, 1978.

\bibitem{Doerfler79:CoversDirectProd}
W.~D\"{o}rfler.
\newblock Multiple {C}overs of {H}ypergraphs.
\newblock {\em Annals of the New York Academy of Sciences}, 319(1):169--176,
  1979.

\bibitem{Doerfler82:DirectProd}
W.~D{\"o}rfler.
\newblock {On the direct product of hypergraphs.}
\newblock {\em Ars Comb.}, 14:67--78, 1982.

\bibitem{DoerflerImrich70:Xsum}
W.~D{\"o}rfler and W.~Imrich.
\newblock {\"Uber die X-Summe von Mengensystemen.}
\newblock {Combinat. Theory Appl., Colloquia Math. Soc. Janos Bolyai 4, 297-309
  (1970).}, 1970.

\bibitem{Doerfler80:CategoriesHpg}
W.~D\"{o}rfler and D.A. Waller.
\newblock A category-theoretical approach to hypergraphs.
\newblock {\em Arch. Math.}, 34:185--192, 1980.

\bibitem{Fontana:98a}
W.~Fontana and P.~Schuster.
\newblock Continuity in {E}volution: {O}n the {N}ature of {T}ransitions.
\newblock {\em Science}, 280:1451--1455, 1998.

\bibitem{Fontana:98b}
W.~Fontana and P.~Schuster.
\newblock Shaping {S}pace: {T}he {P}ossible and the {A}ttainable in {RNA}
  {G}enotype-{P}henotype {M}apping.
\newblock {\em J. Theor. Biol.}, 194:491--515, 1998.

\bibitem{Fuerdi88:Matchings}
Z.~F{\"u}redi.
\newblock {Matchings and covers in hypergraphs.}
\newblock {\em Graphs Comb.}, 4(2):115--206, 1988.

\bibitem{Gallo:98}
G.~Gallo and M.~Scutell{\`a}.
\newblock Directed hypergraphs as a modelling paradigm.
\newblock {\em Decisions in Economics and Finance}, 21:97--123, 1998.

\bibitem{GasztImrich71:lexico-englAbstract}
G.~Gaszt and W.~Imrich.
\newblock On the lexicographic and costrong product of set systems.
\newblock {\em Aequationes Mathematicae}, 6:319--320, 1971.

\bibitem{GasztImrich71:lexico}
G.~Gaszt and W.~Imrich.
\newblock {\"Uber das lexikographische und das kostarke Produkt von
  Mengensystemen. (On the lexicographic and the costrong product of set
  systems).}
\newblock {\em Aequationes Math.}, 7:82--93, 1971.

\bibitem{Gringmann:thesis}
L.~Gringmann.
\newblock {\em Hypergraph Products}.
\newblock Diploma thesis, Fakult{\"a}t f{\"u}r Mathematik und Informatik,
  Universit{\"a}t Leipzig, 2010.

\bibitem{Hahn80:DiHpg}
G.~Hahn.
\newblock {\em Directed hypergraphs: the group of their composition}.
\newblock Ph.d. thesis, McMaster University, 1980.

\bibitem{Hahn81:lexico}
G.~Hahn.
\newblock {The automorphism group of a product of hypergraphs.}
\newblock {\em J. Comb. Theory, Ser. B}, 30:276--281, 1981.

\bibitem{Ham-09}
R.~Hammack.
\newblock {On direct product cancellation of graphs.}
\newblock {\em Discrete Math.}, 309(8):2538--2543, 2009.

\bibitem{Hammack:2011a}
R.~Hammack, W.~Imrich, and S.~Klav{\v{z}}ar.
\newblock {\em Handbook of Product Graphs}.
\newblock Discrete Mathematics and its Applications. CRC Press, 2nd edition,
  2011.

\bibitem{HHOS:16}
R.H. Hammack, M.~Hellmuth, L.~Ostermeier, and P.F. Stadler.
\newblock Associativity and non-associativity of some hypergraph products.
\newblock {\em Math. Comp. Sci}, 10(3):403--408, 2016.

\bibitem{HJH+10}
C.~Heine, S.~Jaenicke, M.~Hellmuth, P.F. Stadler, and G.~Scheuermann.
\newblock Visualization of graph products.
\newblock {\em IEEE Transactions on Visualization and Computer Graphics},
  16(6):1082--1089, 2010.

\bibitem{Hel-11}
M.~Hellmuth.
\newblock A local prime factor decomposition algorithm.
\newblock {\em Discrete Mathematics}, 311(12):944 -- 965, 2011.

\bibitem{HL:16}
M.~Hellmuth and F.~Lehner.
\newblock Fast factorization of {C}artesian products of (directed) hypergraphs.
\newblock {\em J. Theor. Comp. Sci.}, 615:1--11, 2016.

\bibitem{recursiveSurvey11}
M.~Hellmuth, L.~Ostermeier, and P.F. Stadler.
\newblock A survey on hypergraph products.
\newblock {\em Math. Comput. Sci}, 6:1--32, 2012.

\bibitem{Imrich67:Mengensysteme}
W.~Imrich.
\newblock Kartesisches {P}rodukt von {M}engensystemen und {G}raphen.
\newblock {\em Studia Sci. Math. Hungar.}, 2:285 -- 290, 1967.

\bibitem{Imrich70:SchwachKartProd}
W.~Imrich.
\newblock \"uber das schwache {K}artesische {P}rodukt von {G}raphen.
\newblock {\em Journal of Combinatorial Theory}, 11(1):1--16, 1971.

\bibitem{IMIZ-75}
W.~Imrich and H.~Izbicki.
\newblock Associative products of graphs.
\newblock {\em Monatshefte f{\"u}r Mathematik}, 80(4):277--281, 1975.

\bibitem{IMKL-00}
W.~Imrich and S.~Klav{\v{z}}ar.
\newblock {\em Product graphs}.
\newblock Wiley-Interscience Series in Discrete Mathematics and Optimization.
  Wiley-Interscience, New York, 2000.

\bibitem{IMKLDO-08}
W.~Imrich, S.~Klav{\v{z}}ar, and F.~R. Douglas.
\newblock {\em {T}opics in {G}raph {T}heory: {G}raphs and Their {C}artesian
  {P}roduct}.
\newblock AK Peters, Ltd., Wellesley, MA, 2008.

\bibitem{Imrich07:linear}
W.~Imrich and I.~Peterin.
\newblock Recognizing cartesian products in linear time.
\newblock {\em Discrete Math.}, 307(3-5):472--483, 2007.

\bibitem{Imrich:06a}
W.~Imrich and P.~F. Stadler.
\newblock A prime factor theorem for a generalized direct product.
\newblock {\em Discussiones Math.\ Graph Th.}, 26:135--140, 2006.

\bibitem{NesestrilRoedl83:SquareProduct}
{J. Ne\v{s}et\v{r}il and V. R{\"o}dl}.
\newblock Products of graphs and their applications.
\newblock In {\em Graph Theory}, volume 1018 of {\em Lecture Notes in
  Mathematics}, pages 151--160. Springer Berlin / Heidelberg, 1983.

\bibitem{KK-08}
A.~Kaveh and K.~Koohestani.
\newblock Graph products for configuration processing of space structures.
\newblock {\em Comput. Struct.}, 86(11-12):1219--1231, 2008.

\bibitem{KR-04}
A.~Kaveh and H.~Rahami.
\newblock An efficient method for decomposition of regular structures using
  graph products.
\newblock {\em Intern. J. Numer. Meth. Eng.}, 61(11):1797--1808, 2004.

\bibitem{McEP-1971}
R.~J. McEliece and E.~C. Posner.
\newblock Hide and seek, data storage, and entropy.
\newblock {\em The Annals of Mathematical Statistics}, 42(5):1706--1716, 1971.

\bibitem{Mubayi07:DirectHpg}
D.~Mubayi and V.~R\"odl.
\newblock {On the chromatic number and independence number of hypergraph
  products}.
\newblock {\em J. Comb. Theory, Ser. B}, 97(1):151--155, 2007.

\bibitem{OstHellmStad11:CartProd}
L.~Ostermeier, M.~Hellmuth, and P.~F. Stadler.
\newblock The {C}artesian product of hypergraphs.
\newblock {\em Journal of Graph Theory}, 2011.

\bibitem{OHK+09}
P.-J. Ostermeier, M.~Hellmuth, K.~Klemm, J.~Leydold, and P.F. Stadler.
\newblock A note on quasi-robust cycle bases.
\newblock {\em Ars Math. Contemp.}, 2(2):231--240, 2009.

\bibitem{Pemantle92:DirectHpg}
R.~Pemantle, J.~Propp, and D.~Ullman.
\newblock {On tensor powers of integer programs}.
\newblock {\em SIAM J. Discrete Math.}, 5(1):127--143, 1992.

\bibitem{Sabidussi60:GraphMult}
G.~Sabidussi.
\newblock Graph {M}ultiplication.
\newblock {\em Mathematische Zeitschrift}, 72(1):446--457, 1960.

\bibitem{Sonntag89:HamCart}
M.~Sonntag.
\newblock {Hamiltonian properties of the Cartesian sum of hypergraphs.}
\newblock {\em J. Inf. Process. Cybern.}, 25(3):87--100, 1989.

\bibitem{Sonntag90:NormalProd}
M.~Sonntag.
\newblock {Hamiltonicity of the normal product of hypergraphs.}
\newblock {\em J. Inf. Process. Cybern.}, 26(7):415--433, 1990.

\bibitem{Sonntag91:Corrigendum}
M.~Sonntag.
\newblock {Corrigendum to: ``Hamiltonicity of the normal product of
  hypergraphs".}
\newblock {\em J. Inf. Process. Cybern.}, 27(7):385--386, 1991.

\bibitem{Sonntag91:HamTraceLexico}
M.~Sonntag.
\newblock {Hamiltonicity and traceability of the lexicographic product of
  hypergraphs.}
\newblock {\em J. Inf. Process. Cybern.}, 27(5-6):289--301, 1991.

\bibitem{Sonntag:thesis}
M.~Sonntag.
\newblock {\em Hamiltonsche {E}igenschaften von {P}rodukten von
  {H}ypergraphen}.
\newblock Habilitation, Fakult{\"a}t f{\"u}r Mathematik und
  Naturwissenschaften, Bergakademie Freiberg, 1991.

\bibitem{Sonntag92:Hamiltonicity}
M.~Sonntag.
\newblock {Hamiltonicity of products of hypergraphs.}
\newblock {Combinatorics, graphs and complexity, Proc. 4th Czech. Symp.,
  Prachatice/Czech. 1990, Ann. Discrete Math. 51, 329-332 (1992).}, 1992.

\bibitem{Sonntag93:Disjunction}
M.~Sonntag.
\newblock {Hamiltonicity of the disjunction of two hypergraphs.}
\newblock {\em J. Inf. Process. Cybern.}, 29(3):193--205, 1993.

\bibitem{BMRStadler:01a}
B.~M.~R. Stadler, P.~F. Stadler, G.~P. Wagner, and W.~Fontana.
\newblock The topology of the possible: Formal spaces underlying patterns of
  evolutionary change.
\newblock {\em J.\ Theor.\ Biol.}, 213:241--274, 2001.

\bibitem{Sterboul74:ChromDirect}
F.~Sterboul.
\newblock {On the chromatic number of the direct product of hypergraphs.}
\newblock {Proc. 1rst Working Sem. Hypergraphs, Columbus 1972, Lect. Notes
  Math. 411, 165-174 (1974).}, 1974.

\bibitem{Wagner:03a}
G.~Wagner and P.~F. Stadler.
\newblock Quasi-independence, homology and the unity of type: A topological
  theory of characters.
\newblock {\em J.\ Theor.\ Biol.}, 220:505--527, 2003.

\bibitem{Zeigarnik:00a}
A.~V. Zeigarnik.
\newblock On hypercycles and hypercircuits in hypergraphs.
\newblock In P.~Hansen, P.~W. Fowler, and M.~Zheng, editors, {\em Discrete
  Mathematical Chemistry}, volume~51 of {\em DIMACS series in discrete
  mathematics and theoretical computer science}, pages 377--383. American
  Mathematical Society, Providence, RI, 2000.

\bibitem{Zhu92:DirectProd}
X.~Zhu.
\newblock On the chromatic number of the product of hypergraphs.
\newblock {\em Ars Comb.}, 34:25--31, 1992.

\end{thebibliography}

\end{document}